\documentclass[11pt]{article}

\advance \topmargin by -\headheight
\advance \topmargin by -\headsep

\evensidemargin \oddsidemargin
\usepackage{amssymb}
\usepackage{amsthm}
\usepackage{amsmath}
\usepackage[american]{babel}
\usepackage{graphicx}
\usepackage[utf8x]{inputenc}
\usepackage{dsfont}
\usepackage{todonotes}
\usepackage{xspace}
\usepackage{url}
\usepackage{wrapfig}
\usepackage{subcaption}
\usepackage{graphicx}
\usepackage{caption}
\usepackage[pdfborder={0 0 0}]{hyperref}

\usepackage{paralist} 

\newcommand{\xtimes}[2][1]{$#1\times #2$}
\newcommand{\noinbf}[1]{\noindent \textbf{#1}}

\newcommand{\email}[1]{\href{mailto:#1}{#1}}

\newcommand{\ignoreforced}[1]{}
\newcommand{\newtext}[1]{\color{black}#1\color{black}\xspace}
\newcommand{\newnewtext}[1]{\color{black}#1\color{black}\xspace}
\newcommand{\new}[1]{\color{black}#1\color{black}\xspace}
\newcommand{\newest}[1]{\color{red}#1\color{black}\xspace}

\newtheorem{theorem}{Theorem}[section]
\newtheorem{lemma}[theorem]{Lemma}
\newtheorem{corollary}[theorem]{Corollary}

\newtheorem{definition}[theorem]{Definition}

\def\compactify{\itemsep=0pt \topsep=0pt \partopsep=0pt \parsep=0pt}
\let\latexusecounter=\usecounter
\newenvironment{itemize*}
  {\def\usecounter{\compactify\latexusecounter}
   \begin{itemize}}
  {\end{itemize}\let\usecounter=\latexusecounter}
\newenvironment{enumerate*}
  {\def\usecounter{\compactify\latexusecounter}
   \begin{enumerate}}
  {\end{enumerate}\let\usecounter=\latexusecounter}
\newenvironment{description*}
  {\begin{description}\compactify}
  {\end{description}}

\def\nullglue{{\rm null}}
\def\emptytile{{\rm empty}}
\def\north{\mathop{\rm north}\nolimits}
\def\south{\mathop{\rm south}\nolimits}
\def\west{\mathop{\rm west}\nolimits}
\def\east{\mathop{\rm east}\nolimits}

\title{New Geometric Algorithms \\ for Fully Connected Staged Self-Assembly\thanks{A preliminary extended abstract appears in the Proceedings of DNA'21, 2015~\cite{dfs+-ngafcssa-15}.}}
\author{
   Erik D. Demaine\thanks{
   	CSAIL, MIT, USA. \email{edemaine@mit.edu}}
   \and
   S\'{a}ndor P. Fekete\thanks{
   Department of Computer Science, TU Braunschweig, Germany. \email{s.fekete@tu-bs.de}, \email{c.scheffer@tu-bs.de}, \email{arne.schmidt@tu-bs.de}}
   \and Christian Scheffer$^\ddagger$
   \and Arne Schmidt$^\ddagger$
}

\date{}

\begin{document}

\maketitle

\begin{abstract}
We consider \textit{staged self-assembly systems},
in which square-shaped tiles can be added to bins in several stages. Within these
bins, the tiles may connect to each other, depending on the {\em glue types} of their edges. 
Previous work by Demaine et al.\ 
showed that a relatively small number of tile types
suffices to produce arbitrary shapes in this model. However, these constructions were only
based on a spanning tree of the geometric shape, so they did not produce
full connectivity of the underlying grid graph in the case of shapes with holes; 
self-assembly of fully connected assemblies with a polylogarithmic number of stages was left as a major open problem.
We resolve this challenge by presenting new systems for staged assembly that
produce fully connected polyominoes in $\mathcal{O}(\log^2 n)$ stages, for various scale factors and
temperature $\tau=2$ as well as $\tau=1$.
Our constructions work even for shapes with holes and \new{use} only a constant number
of glues and tiles. Moreover, the underlying approach is more geometric in nature, implying
that it promises to be more feasible for shapes with compact geometric description.
\end{abstract}



\section{Introduction}
In {\em self-assembly}, a set of simple {\em tiles} form complex structures without
any active or deliberate handling of individual components. Instead, the overall construction
is governed by a simple set of rules, which describe how mixing the tiles leads
to bonding between them and eventually a geometric shape.

The classic theoretical model for self-assembly is the 
\textit{abstract tile-assembly model} (aTAM).
It was first introduced by Winfree~\cite{winfree1998,Rot00}. The
\textit{tiles} used in this model are building blocks, 
which are unrotatable squares with a specific glue on
each side. Equal glues have a connection strength and may stick together. 
The \textit{glue complexity} of a tile set $T$ is the number of different glues \newnewtext{on all the tiles in $T$},
while the \textit{tile complexity} of $T$ is the number of different tile types in $T$.
If an additional tile wants to attach to the existing assembly
by making use of matching glues,
the sum of corresponding glue strengths 
needs to be at least some minimum value $\tau$, which is called the
\textit{temperature}.  

A generalization of the aTAM called the {\em two-handed assembly model} (2HAM) was introduced by Demaine et al.~\cite{DDF08}. 
While in the aTAM, only individual tiles can be attached to an existing intermediate assembly, the 2HAM
allows attaching other partial assemblies.
If two partial assemblies (``supertiles'') want to assemble, then 
the sum of the glue strength along the whole common boundary
needs to be at least $\tau$. 

In this paper we consider the \textit{staged tile assembly model}
introduced in~\cite{DDF08}, which is based on the 2HAM. In
this model the assembly process is split into sequential stages that are kept in separate bins, with
supertiles from earlier stages mixed together consecutively to gain new
supertiles. We can either add a new tile to an existing bin, or we pour one bin into
another bin, such that the content of both \new{gets} mixed; afterwards, unassembled parts get removed. The overall number 
of stages and bins \newnewtext{of a system} are the \textit{stage complexity} and the \textit{bin complexity}. 
Demaine et al.~\cite{DDF08} achieved several results
summarized in Table~\ref{overview}. Most notably, they presented a
system (based on a spanning tree) that can produce arbitrary polyomino shapes $P$ in $\mathcal{O}(\mbox{\em diameter})$
many stages, $\mathcal{O}(\log N)=\mathcal{O}(\log n)$ bins and a constant number of glues, where 
$N$ is the number of \new{unit squares, called {\em pixels}, whose union
forms  $P$}, $n$ is the size of \new{the bounding box}, i.e., a smallest square
containing $P$, and the diameter is measured \new{by the maximum length of a
shortest path between any two pixels in the adjacency graph of the
pixels in $P$; this} can be as big as $N$. The downside is that the resulting
\new{supertiles} are not fully connected. For achieving full connectivity, only the special case of monotone shapes was resolved
by a system with $\mathcal{O}(\log n)$ stages; for hole-free shapes, \new{Demaine et al.~\cite{DDF08}} were able to give a system
with full connectivity, scale factor 2, but $\mathcal{O}(n)$ stages. This left a major open problem: designing a staged assembly system
with full connectivity, polylogarithmic stage complexity and constant scale factor for general shapes.

{\bf Our results.}
We show that for any polyomino, even with holes, there is a staged assembly system with the following properties, both
for $\tau=2$ and $\tau=1$.
\begin{enumerate}
\item polylogarithmic stage complexity,
\item constant glue and tile complexity,
\item constant scale factor,
\item full connectivity.
\end{enumerate} 
See Table~\ref{overview} for an overview. The main novelty of our method is to focus on the underlying geometry
of a constructed shape $P$, instead of just its connectivity graph. This results in \newtext{bin complexities} that are a function
of $k$, the number of vertices of $P$: while $k$ can be as big as $\Theta(n^2)$, 
$n$ can be arbitrarily large for fixed $k$, implying that our approach promises to be more suitable
for constructing natural shapes with a clear geometric structure.

\begin{table*}[t!]
\centering
\fontsize{7.0}{14} \selectfont
\begin{tabular}{l|@{\,}c@{\,}|@{\,}c@{\,}|@{\,}c@{\,}|@{\,}c@{\,}|@{\,}c@{\,}|@{\,}c@{\,}|@{\,}c@{\,}|@{\,}c@{\,}}
\textbf{Lines and Squares} & \textbf{Glues} & \textbf{Tiles} & \textbf{Bins} & \textbf{Stages} & \textbf{$\tau$} & \textbf{Scale} & \textbf{Conn.} & \textbf{Planar}\\
\hline
Line \cite{DDF08} & 3 & 6 & 7 & $\mathcal{O}(\log n)$ & 1 & 1 & full & yes\\
\hline
Square --- Jigsaw techn. \cite{DDF08} & 9 & $\mathcal{O}(1)$ & $\mathcal{O}(1)$ & $\mathcal{O}(\log n)$ & 1 & 1  & full & yes\\
\hline
Square --- $\tau = 2$ {\bf (Sect.~\ref{tsq})} & 4 & $\mathcal{O}(1)$ & $\mathcal{O}(1)$ & $\mathcal{O}(\log n)$ & 2 & 1 & full & yes\\ 
\hline
\multicolumn{9}{c}{}\\

\textbf{Arbitrary Shapes} & \textbf{Glues} & \textbf{Tiles} & \textbf{Bins} & \textbf{Stages} & \textbf{$\tau$} & \textbf{Scale} & \textbf{Conn.} & \textbf{Planar}\\
\hline
Spanning Tree Method \cite{DDF08} & 2 & 16 & $\mathcal{O}(\log n)$ & $\mathcal{O}(diameter)$ & 1 & 1 &partial & no\\
\hline
Monotone Shapes \cite{DDF08} & 9 & $\mathcal{O}(1)$ & $\mathcal{O}(n)$ & $\mathcal{O}(\log n)$ & 1 & 1 & full & yes\\
\hline
Hole-Free Shapes \cite{DDF08} & 8 & $\mathcal{O}(1)$ & $\mathcal{O}(N)$ & $\mathcal{O}(N)$ & 1 & 2 & full & no\\
\hline
Shape with holes {\bf (Sect.~\ref{holeshape2})} & $\new{6}$ & $\mathcal{O}(1)$ & $\mathcal{O}(k)$ & $\mathcal{O}(\log^2 n)$ & 2 & 3 & full & no\\
\hline
Hole-Free Shapes {\bf (Sect.~\ref{holeshape2})} & $\new{6}$ & $\mathcal{O}(1)$ & $\mathcal{O}(k)$ & $\mathcal{O}(\log n)$ & 2 & 3 & full & no\\
\hline
Hole-Free Shapes {\bf (Sect.~\ref{holefreeshape})} & $18$ & $\mathcal{O}(1)$ & $\mathcal{O}(k)$ & $\mathcal{O}(\log^2 n)$ & 1 & 4 & full & no\\
\hline
Shape with holes {\bf (Sect.~\ref{holeshape1})} & $20$ & $\mathcal{O}(1)$ & $\mathcal{O}(k)$ & $\mathcal{O}(\log^2 n)$ & 1 & 6 & full & no\\
\hline
\end{tabular}
\vspace*{2mm}

\caption[Result overview]{Overview of results from \cite{DDF08} and this paper.
The number of \newtext{pixels} of $P$ is denoted by $N\in\mathcal{O}(n^2)$, $n$ is the
side length of a smallest bounding square, while $k$ is the number of vertices
of the polyomino, with $k\in\Omega(1)$ and $k\in\mathcal{O}(N)$. \newtext{The {\em diameter} is the maximum of all shortest paths between any two pixels in the adjacency graph of the corresponding shape.}}
\label{overview}
\end{table*}

{\bf Related work.}
As mentioned above, our work is based on the 2HAM. There is a variety
of other models, e.g., see \cite{aggarwal2005complexities}.
A variation of the staged 2HAM is the {\em Staged Replication Assembly Model} 
by Abel et al.~\cite{abel2010shape}, which aims at reproducing supertiles by using
\textit{enzyme self assembly}. \
Another variant is the {\em \new{Signal-passing} Tile Assembly Model} introduced by Padilla et al.~\cite{padilla2014asynchronous}.

Other related geometric work by Cannon et al.~\cite{DBLP:journals/corr/abs-1201-1650} 
and Demaine et al.~\cite{DBLP:journals/corr/abs-1212-4756} considers
reductions between different systems, often based on geometric properties.
Fu et al.~\cite{fu2012self} use geometric tiles in a generalized tile assembly model to
assemble shapes. Fekete et al.~\cite{f_new} study the power
of using more complicated polyominoes as tiles.

\ignoreforced{
Kao and Schweller~\cite{kao06} and Summers~\cite{DBLP:journals/corr/abs-0907-1307} use a sequence of different
temperatures called \textit{temperature programming} to assemble objects.
Keenan et al.~\cite{DBLP:journals/corr/abs-1303-2416} and Brun~\cite{brun2007arithmetic}
Randomized aspects of composition were considered by Doty~\cite{doty2010randomized}
and Kao and Schweller~\cite{kao2008randomized}.
}

Using stages has also received attention in DNA self assembly.
Reif~\cite{reif1999local} uses a stepwise model for parallel computing.
Park et al.~\cite{park2006finite} consider assembly techniques with hierarchies to assemble DNA
lattices. Somei et al.~\cite{somei2006microfluidic} use a stepwise assembly
of DNA tiles. {Padilla et al.~\cite{padilla2012hierarchical} include active signaling and glue activation in the aTAM to control hierarchical assembly of Robinson patterns.} None of these works considers complexity aspects.

\section{The Staged Assembly Model}
\label{sec:model}
\new{In this section, we present basic definitions common to most 
assembly models, followed by a description of the staged assembly model, and finally we
define various metrics to measure the efficiency of a staged assembly system.
Staged assembly systems were introduced by Demaine et al.~\cite{DDF08}. As we use
their basic definitions, we quote the corresponding framework (mostly verbatim)
for self-containedness of this paper. }

\new{\paragraph{Polyominoes.}
A polyomino $P$ is a polygon with axis-parallel edges of integer length.
(See Figure~\ref{fig:polyomino}.)
A vertex of $P$ is a {\em reflex vertex}, if its interior angle is $3\pi/2$.
Every polyomino can be decomposed into a set of unit squares, called \emph{pixels};
without loss of generality, we assume they are centered at points from $\mathbb{Z}^2$;
as described below, this means they correspond to tile positions in a configuration.
The \emph{degree} of a pixel is
the number of (vertically or horizontally) adjacent pixels in the polyomino. A
pixel is a {\em boundary pixel} if at least one of the eight (axis parallel or diagonal) neighbor
positions is not occupied by a pixel in $P$. A boundary pixel $p$ is an {\em ordinary} boundary pixel
if precisely two (say, $p'$ and $p''$) of its four vertical or horizontal neighbor pixels in $P$ are boundary pixels, 
and the positions of $p$, $p'$, $p''$ are collinear; a boundary pixel is a 
{\em corner} pixel of $P$ if is not an ordinary boundary pixel. 
The {\em corner set of a polyomino} is the set of all
corner pixels along all boundaries.} 

\begin{figure}[ht]
  \begin{center}
       \includegraphics[height=6cm]{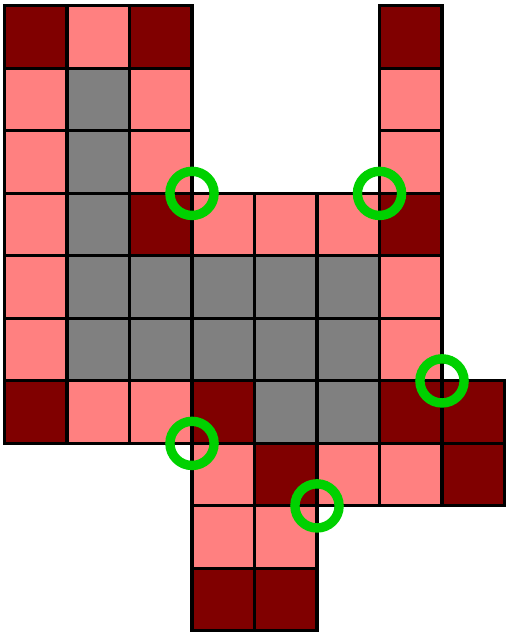}
  \end{center}
  \vspace*{-12pt}
  \caption{A polyomino, its reflex vertices (indicated by green circles), its corner pixels (dark red) and its ordinary
boundary pixels (light red).}
  \label{fig:polyomino}
\end{figure}

\paragraph{Tiles and tile systems.} A \emph{(Wang) tile} $t$ is a \new{non-rotatable} unit square
defined by the ordered quadruple
$\langle \north(t), \east(t), \south(t), \west(t) \rangle$
of glues on the four edges, also called \emph{sides}, of the tile.
Each \emph{glue} is taken from a finite alphabet $\Sigma$,
which includes a special ``null'' glue denoted {\em null}.
For simplicity of bounds, we do not count the $\nullglue$ glue
in the \emph{glue complexity} $g=|\Sigma|-1$.

A \emph{tile system} is an ordered triple $\langle T, G, \tau \rangle$
consisting of the \emph{tileset} $T$ (a set of distinct tiles),
the \emph{glue function} $G: \Sigma^2 \rightarrow \{0, 1, \dots, \tau\}$,
and the \emph{temperature}~$\tau$ (a positive integer).
It is assumed that $G(x, y) = G(y, x)$ for all $x, y \in \Sigma$ and that
$G(\nullglue,x) = 0$ for all $x \in \Sigma$.
Indeed, in all of our constructions $G(x,y) = 0$ for all $x \neq y$,
and each $G(x,x) \in \{1, 2, \dots, \tau\}$.
The \emph{tile complexity} of the system is $|T|$.
\vspace{-.1in}
\paragraph{Configurations.}
A \emph{configuration}\new{, which we also call an \emph{assembly},} is a function
$C: \mathbb{Z}^2 \rightarrow T \cup \{\emptytile\}$,
where $\emptytile$ is a special tile that has the $\nullglue$
glue on each of its four edges. The {\em shape} of a configuration $C$ is
the set of positions $(i, j)$ that do not map to the $\emptytile$ tile.
The shape of a configuration can be disconnected, corresponding to
several distinct supertiles. 

\paragraph{Adjacency graph and supertiles.}
Define the \emph{adjacency graph} $G_C$ of a configuration $C$ as
follows.  The vertices are coordinates $(i, j)$ such that
$C(i, j) \neq \emptytile$.  There is an edge between two vertices
$(x_1, y_1)$ and $(x_2, y_2)$ if and only if $|x_1-x_2|+ |y_1 - y_2| = 1$. A set of vertices $\{(x_1,y_1),..., (x_k,y_k)\}$ are \emph{collinear} if $x_1= \dots x_k$ or $y_1 = \dots = y_k$. A {\em vertex} of a configuration $C$ is a vertex $(i,j)$ of $G_C$ if there are $x,y \in \{ -1,1 \}$ such that $C(i+x,j),C(i,j+y) = \emptytile$ and $C(i-x,j),C(i,j-y) \neq \emptytile$. A vertex $C(i,j)$ of $C$ is \emph{reflex} if $C(i,j)$ is adjacent to one position that is the empty tile.

A \emph{supertile} is a maximal connected subset $G^{\prime}$ of~$G_C$,
i.e., $G^{\prime} \subseteq G_C$ such that, for every connected
subset~$H$, if $G^{\prime} \subseteq H \subseteq G_C$, then $H = G^{\prime}$.
For a supertile $S$, let $|S|$ denote the number of nonempty positions
(tiles) in the supertile. We call $|S|$ the \emph{size} of $S$.
Throughout this paper, we informally refer to (lone) tiles as
a special case of supertiles.



If every two adjacent tiles in a supertile share a positive strength
glue type on abutting edges, the supertile is \emph{fully
connected}. \new{All provided assembly systems in this paper are fully connected.}


\vspace{-.1in}
\paragraph{Staged assembly systems.}  For any two supertiles $X$ and $Y$, the \emph{combination}
set $C^{\tau}_{(X,Y)}$ of $X$ and $Y$ is defined to be the set of
all supertiles obtainable by placing $X$ and $Y$ adjacent
to each other (without overlapping) such that, if we list each newly
coincident edge $e_i$ with edge strength $s_i$, then $\sum s_i \geq \tau$.

A \emph{bin} is a pair $(S, \tau)$, where $S$ is a set of initial
supertiles whose tile types are contained in a given set of tile types $T$, and $\tau$ is a
temperature parameter.  For a bin $(S, \tau)$, the set of
\emph{produced} supertiles $P'_{(S,\tau)}$ is defined recursively as
follows:  (1)~$S \subseteq P'_{(S,\tau)}$ and (2)~for any $X,Y\in
P'_{(S,\tau)}$, $C^\tau_{(X,Y)} \subseteq P'_{(S,\tau)}$. The set of
\emph{terminally} produced supertiles of a bin $(S,\tau)$ is
$P_{(S,\tau)} = \{ X\in P' \mid Y\in P'$, $C^\tau_{(X,Y)} =
\emptyset \}$.  The set of supertiles $P$ is \emph{uniquely}
produced by bin $(S,\tau)$ if each supertile in $P'$ is of finite
size. 

We can \emph{create} a bin of a single tile type $t \in T$, we can
\emph{merge} multiple bins together into a single bin, and we can
\emph{split} the contents of a given bin into multiple new bins.  In
particular, when splitting the contents of a bin, we assume the
ability to extract only the unique terminally produced set of
supertiles $P$, while filtering out additional partial assemblies in
$P'$. 





An \emph{$r$-stage $b$-bin mix graph} $M$ consists of $r b+1$
vertices, $m_*$ and $m_{i,j}$ for $1\leq i \leq r$ and $1\leq j \leq
b$, and an arbitrary collection of edges of the form $(m_{r,j},m_*)$
or $(m_{i,j}, m_{i+1,k})$ for some $i, j, k$.

	A \emph{staged assembly system} is a $3$-tuple $\langle M_{r,b}, \{T_{i,j}\},
\{\tau_{i,j}\} \rangle$, where $M_{r,b}$ is an $r$-stage $b$-bin mix
graph, each $T_{i,j}$ is a set of tile types, and each $\tau_{i,j}$
is an integer temperature parameter.  Given a staged assembly system,
for each $1\leq i \leq r$, $1\leq j \leq b$, we define a
corresponding bin $(R_{i,j}, \tau_{i,j})$, where $R_{i,j}$ is
defined as follows:
\begin{enumerate*}
\item $R_{1,j}= T_{1,j}$ (this is a bin in the first stage);
\item For $i\geq 2$,
  $\displaystyle R_{i,j}= \Big(\bigcup_{k:\ (m_{i-1,k},m_{i,j})\in M_{r,b}} P_{(R_{(i-1,k)},\tau_{i-1,k})}\Big) \cup T_{i,j}$.
\item $\displaystyle R_* =\bigcup_{k:\ (m_{r,k},m_{*})\in M_{r,b}} P_{(R_{(r,k)},\tau_{r,k)})}$.
\end{enumerate*}

The set of terminally produced supertiles for a staged assembly system are defined as $P_{(R_{*}, \tau_*)}$. A staged assembly system uniquely produces the set of supertiles $P_{(R_{*}, \tau_*)}$ if in each bin the terminal supertiles are unique.

Throughout this paper, we assume that, for all $i,j$,
$\tau_{i,j} = \tau$ for some fixed global temperature~$\tau$,
and we denote a staged assembly system as
$\langle M_{r,b}, \{T_{i,j}\}, \tau \rangle$.

	The following metrics are considered: The \emph{tile complexity $|\bigcup T_{i,j}|$}, the \emph{bin complexity $b$}, the \emph{stage complexity $r$}, and the \emph{temperature $\tau$}. 
	
	A staged assembly system is \emph{planar} if supertiles have obstacle-free paths
  to reach their mates in every possible sequence of attachments. In a \emph{fully connected} supertile, every two adjacent tiles have the same positive-strength glue along their common edge. Otherwise the supertile is \emph{partially connected}.

\section{Fully Connected Constructions for $\tau=2$}\label{t=2}

In the following, we consider fully connected assemblies for temperature $\tau=2$.
We start by an approach for squares (Section~\ref{tsq}).
In Section~\ref{holeshape2} we describe how to extend this basic idea to assembling general polyominoes. 

\subsection{\boldmath $n\times n$ Squares, $\tau = 2$}\label{tsq}
For $\tau = 2$ assembly systems, it is possible to develop more efficient ways for constructing a square.
The construction is based on an idea by Rothemund and Winfree~\cite{Rot00}, which we adapt to staged assembly.
Basically, it consists of connecting two strips 
by a corner tile, before filling up this frame; see Figure~\ref{temp2square}.

\begin{figure}[h!]
\centering
\begin{minipage}{0.20\columnwidth}
\includegraphics[width = \columnwidth]{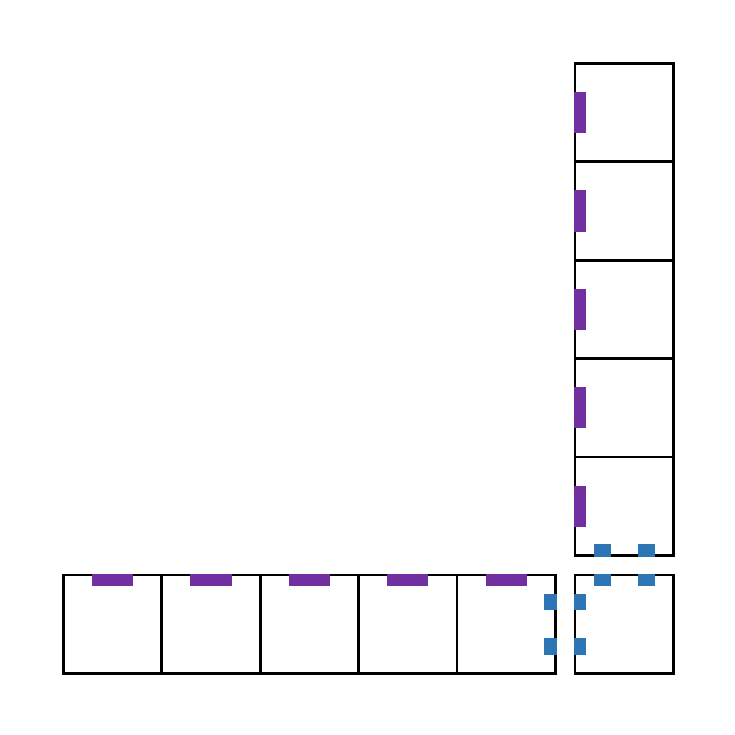}
\end{minipage}
\begin{minipage}{0.20\columnwidth}
\includegraphics[width = \columnwidth]{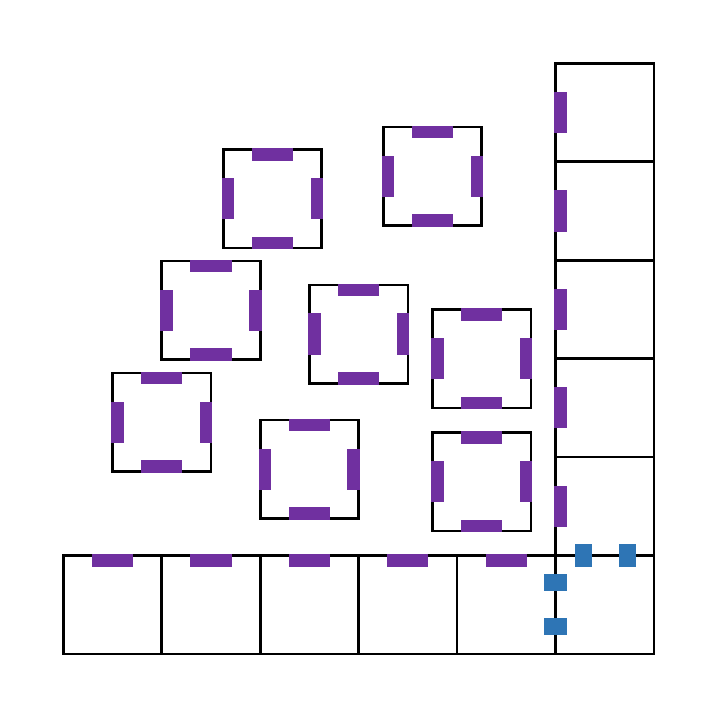}
\end{minipage}
\begin{minipage}{0.20\columnwidth}
\includegraphics[width = \columnwidth]{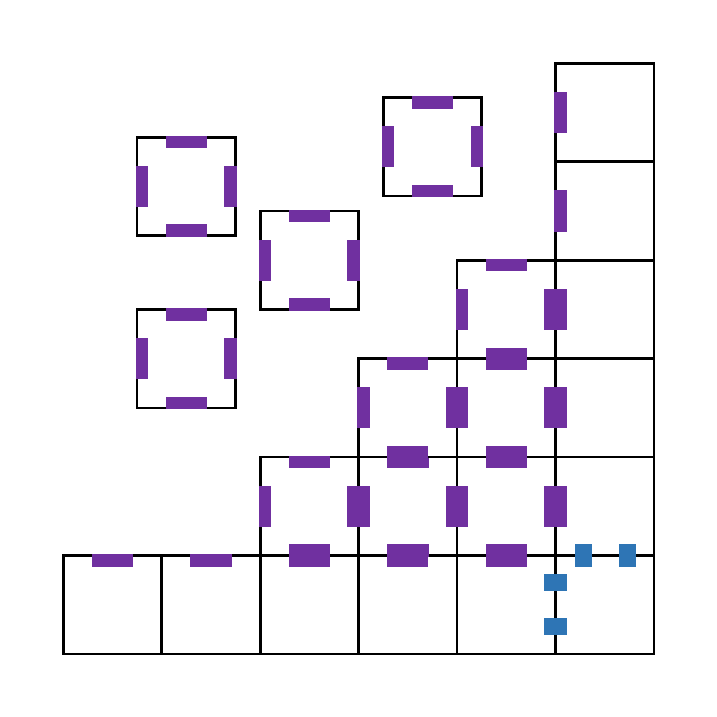}
\end{minipage}
\begin{minipage}{0.20\columnwidth}
\includegraphics[width = \columnwidth]{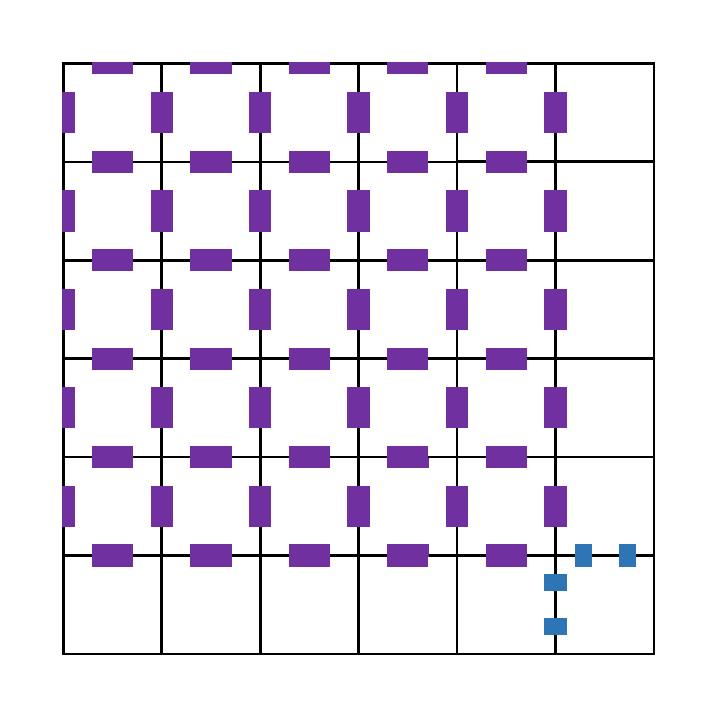}
\end{minipage}
\caption[Square assembly via a corner construction]{Construction of fully connected square using $\tau=2$ and a frame.}
\label{temp2square}
\end{figure}

\begin{theorem}
\label{th:one}
There exists a $\tau=2$ staged assembly system that assembles a fully connected $n \times n$ square with $\mathcal{O}(\log n)$ stages, $4$ glues, $14$ tiles and $7$ bins.
\label{squareTheorem}
\end{theorem}
\begin{proof}
{The construction is an easy result of combining \new{a} known construction for lines by staged assembly with filling in squares in the aTAM with temperature $\tau=2$,
as follows.}
First we construct the \xtimes{(n-1)} strips with strength-2 glues. We know from~\cite{DDF08}
that a strip can be constructed in $\mathcal{O}(\log
n)$ stages, three glues,  six  tiles and seven bins. Because both strips are perpendicular,
they do not connect. Therefore, we can use all seven bins to construct both strips 
in parallel. For each strip we use tiles such that the edge 
toward the interior of the square has a strength-1 glue. In the next stage we mix
the single corner tile with the two strips. Finally, we add a
tile type with strength-1 glues on all sides. When the square is filled, no
further tile can still connect, as $\tau=2$.

Overall, we need $\mathcal{O}(\log n)$ stages with four glues (three for the
construction, one for filling up the square), 14 tiles (six for each of the two strips, one for
the corner tile, one for filling up the square) and seven bins for the parallel
construction of the two strips.  
\end{proof}

\subsection{\boldmath \newtext{Polyominoes} with or without Holes, $\tau = 2$}\label{holeshape2}
Our method for assembling a polyomino $P$ at $\tau =2$ generalizes the 
approach for building a square that is described in Section~\ref{tsq}. The key idea is to scale $P$ by a factor
of 3, yielding $P^3$; for this we first build a frame called the \emph{backbone}, which is a \new{specific} spanning 
tree based on the union of all boundaries of $P^3$. This backbone is then filled up 
in a final stage by applying a more complex version of the flooding approach of Theorem~\ref{squareTheorem}. 
In particular, there is not only one flooding tile, but a constant set $S$ of 
such distinct tiles. 

\subsubsection{\boldmath Definition and construction of the backbone.}\label{defback} 
In the following, we consider a scaled copy $P^3$ of a polyomino $P$, constructed by replacing each
pixel by a $3\times 3$ square of pixels. We define the {\em backbone} of $P^3$ as follows; see
Figure~\ref{backboneC} for an illustration.

\begin{figure}[ht]
  \begin{center}
    \begin{tabular}{ccc}
      \includegraphics[height=4cm]{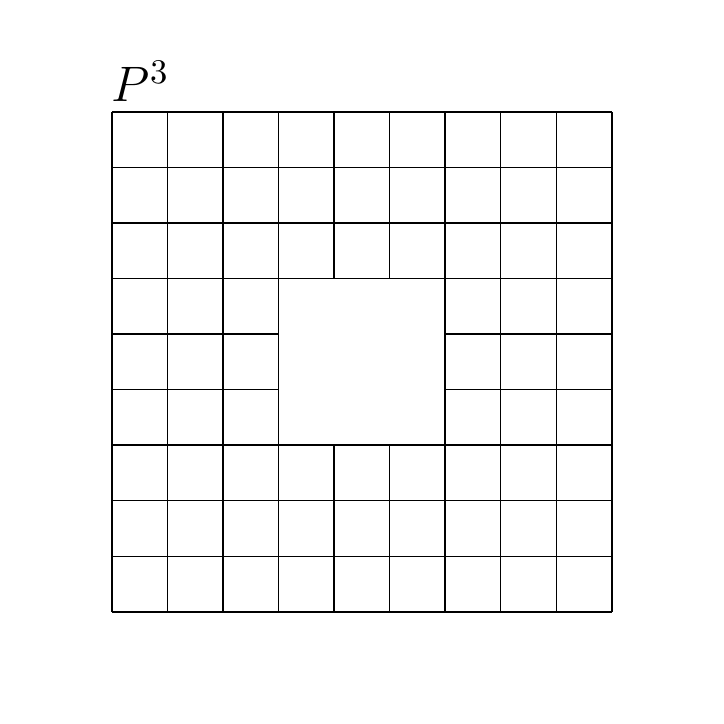} &
      \includegraphics[height=4cm]{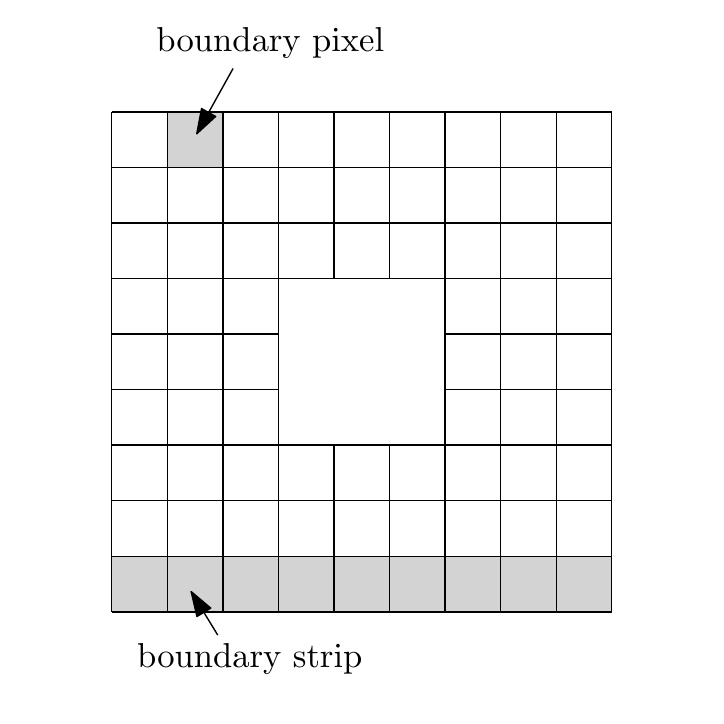}&
      \includegraphics[height=4cm]{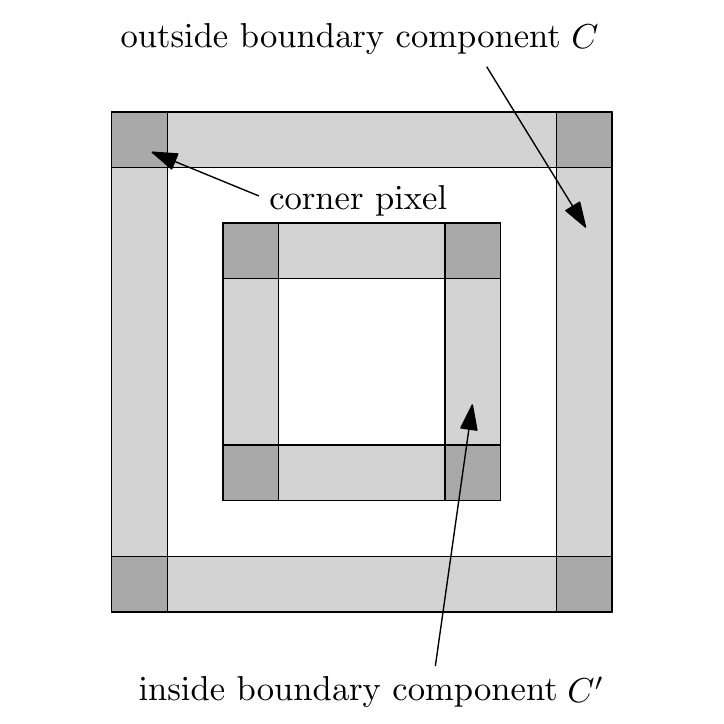}\\[2ex]
      {\small (a) 3-scaled polyomino $P^3$.} &      
      {\small (b) Boundary pixels and strips.} &
      {\small (c) Outside and inside}\\
      {\small } &      
      {\small } &
      {\small boundary components.}\\
      \includegraphics[height=4cm]{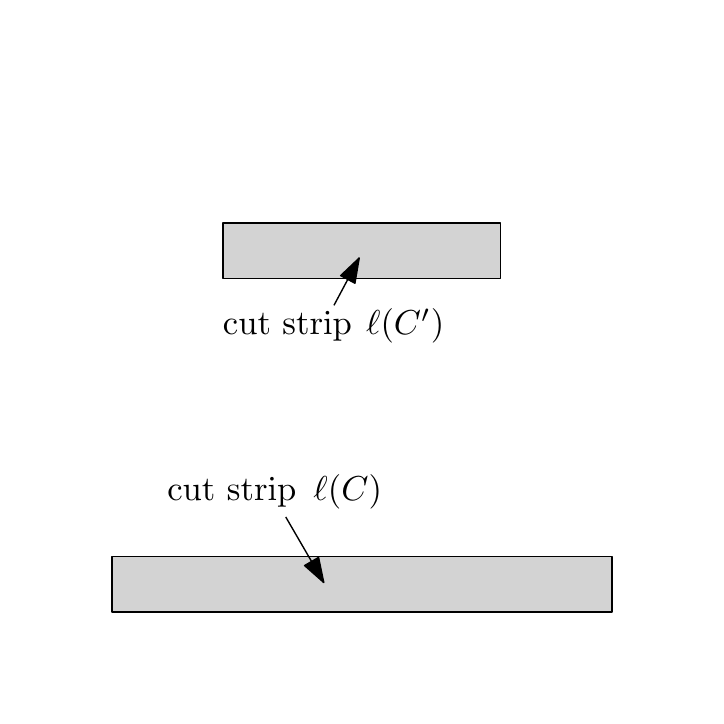} &
      \includegraphics[height=4cm]{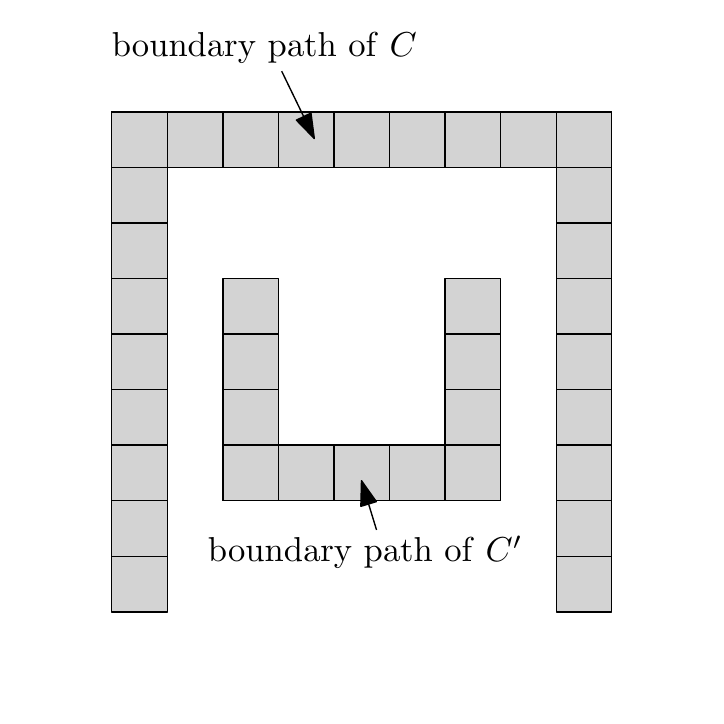}&
      \includegraphics[height=4cm]{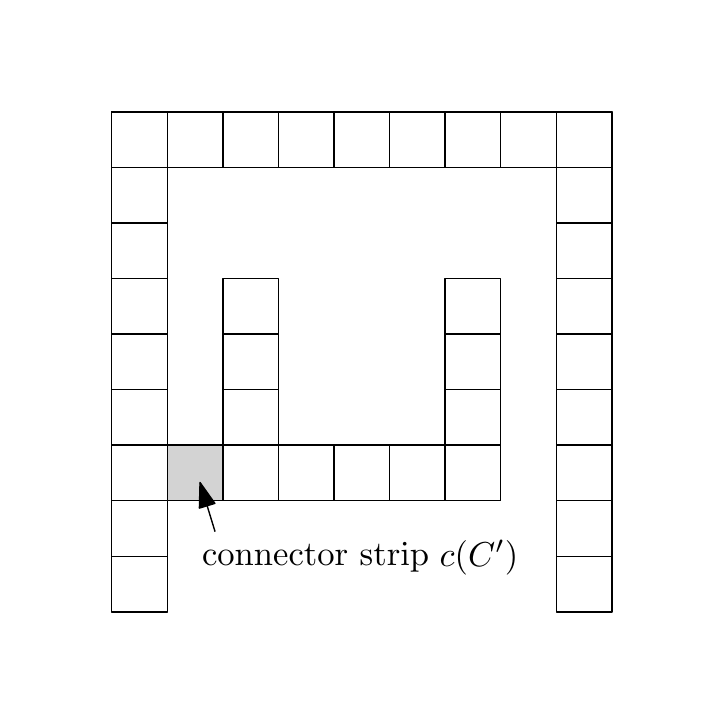}\\[2ex]
      {\small (d) Cut strip of $C$ and $C'$.} &      
      {\small (e) Boundary paths.} &
      {\small (f) Connector strip $c(C')$.}\\
    \end{tabular}
  \end{center}
  \caption{Stepwise construction of the backbone of a 3-scaled polyomino $P^3$.}
  \label{backboneC}
\end{figure}


\newtext{
\begin{definition}\label{def:backbone1}
A pixel of $P^3$ is a {\em boundary pixel} of $P^3$,
if one of the pixels in its eight (axis-parallel or diagonal) neighbor pixels does not
belong to $P^3$. A \emph{boundary strip} of $P^3$ is a maximal set of
boundary pixels that forms a contiguous (vertical or horizontal) strip, see Figure~\ref{backboneC}(b).
A {\em boundary component} $C$ is a maximal connected component of boundary pixels.
\end{definition}

Because of the scaling, an {\em inside boundary component} corresponds to precisely one inside
boundary of $P^3$ (delimiting a hole), while the {\em outside boundary component} corresponds to
the exterior boundary of $P^3$, see Figure~\ref{backboneC}(c). Furthermore, each boundary component $C$ has a unique
decomposition into boundary strips: a circular sequence of boundary strips that alternate between vertical
and horizontal, with consecutive strips sharing a single (``corner'') pixels.

\begin{definition}\label{def:backbone}
For an inside boundary component $C'$, its {\em cut strip} $\ell(C')$ is the leftmost of its topmost strips;
for the outside boundary component, its cut strip $\ell(C)$ is the leftmost of its bottommost strips, see Figure~\ref{backboneC}(d).

A {\em boundary path} of the outside boundary component $C$ consists of the union of all its strips,
with the exception of $\ell(C)$; for an inside boundary component $C'$, it consists of $C'\setminus\ell(C')$; see Figure~\ref{backboneC}(e).
Furthermore, the {\em connector strip} $c(C')$ for an inside boundary component $C'$ is the 
contiguous horizontal set of pixels of $P^3$ extending to the left from the leftmost bottommost
pixel of $C'$ and ending with the first encountered other boundary pixel of $P^3$; see Figure~\ref{backboneC}(f). {\newnewtext{Note that no tile of the boundary path is part of the connector strip.}}
Then the {\em backbone} of $P^3$ is the union of all boundary paths and the connector strips of inside boundary components.
\end{definition}
}

By construction, the backbone has a canonical decomposition into boundary strips and connector strips;
furthermore, a pixel in the backbone of $P^3$ \newnewtext{has degree three if and only if this pixel is adjacent to a pixel of a connector strip.}
For $h$ holes, only $2h$ pixels in the backbone \newnewtext{can have degree three.} 



Overall, this yields a hole-free shape that can be constructed efficiently.

\begin{lemma}\label{lem:backbone}
Let $k$ be the number of vertices of a $3$-scaled polyomino $P^3$. The corresponding backbone can be assembled in
$\mathcal{O}(\log^2 n)$ stages with \new{$3$} glues, $\mathcal{O}(1)$ tiles and
$\mathcal{O}(k)$ bins.  \end{lemma}

\begin{proof}
	The main idea is to give a \new{hierarchical tree decomposition} of the backbone \newtext{$T$} into
\new{subtrees, all the way down to single pixels. Assembling the backbone is then performed in a bottom-up fashion
from the tree decomposition. To this end, subtrees in the decomposition are split into smaller
pieces by the removal of appropriate pixels. An important invariant is that each subtree is separated from the
rest of the backbone by at most two pixels; when assembling the backbone in a bottom-up fashion
from the tree decomposition, this ensures that a constant number of glues at the separating 
pixels suffices for assembling the whole backbone.} 

\new{For decomposing the backbone, we observe that it consists of two types of components: strips and corner pixels;
see Figure~\ref{fig:backboneStripsandCorners}. By construction of the backbone, the degree of corner pixels is two or three, 
corresponding to the number of adjacent strips.}

\begin{figure}[h!]
\centering
\begin{subfigure}[t]{0.2\columnwidth}
\includegraphics[width = \columnwidth]{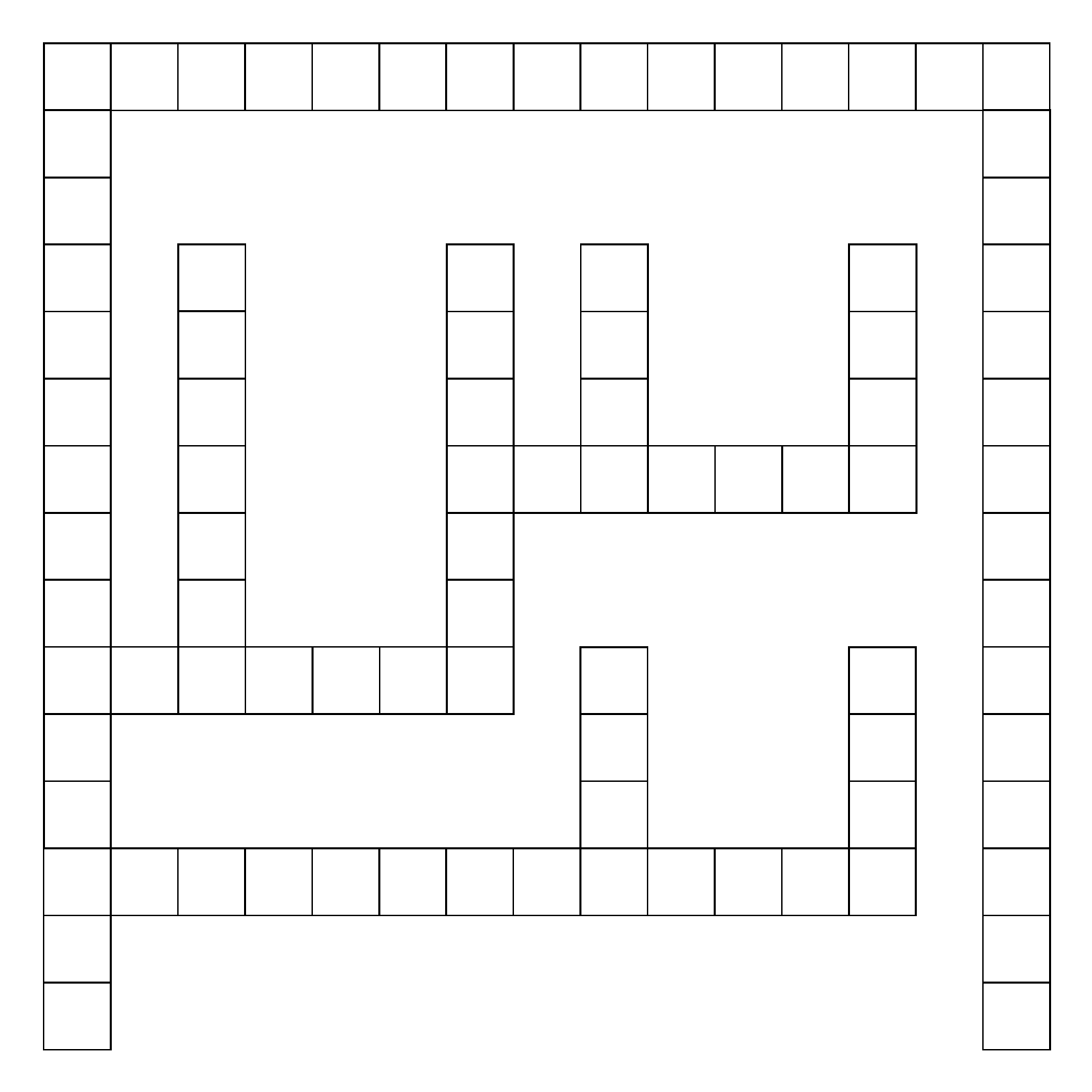}
\caption{\newnewtext{A backbone of a scaled polyomino $P^3$.}}
\end{subfigure}
\hfil
\begin{subfigure}[t]{0.2\columnwidth}
\includegraphics[width = \columnwidth]{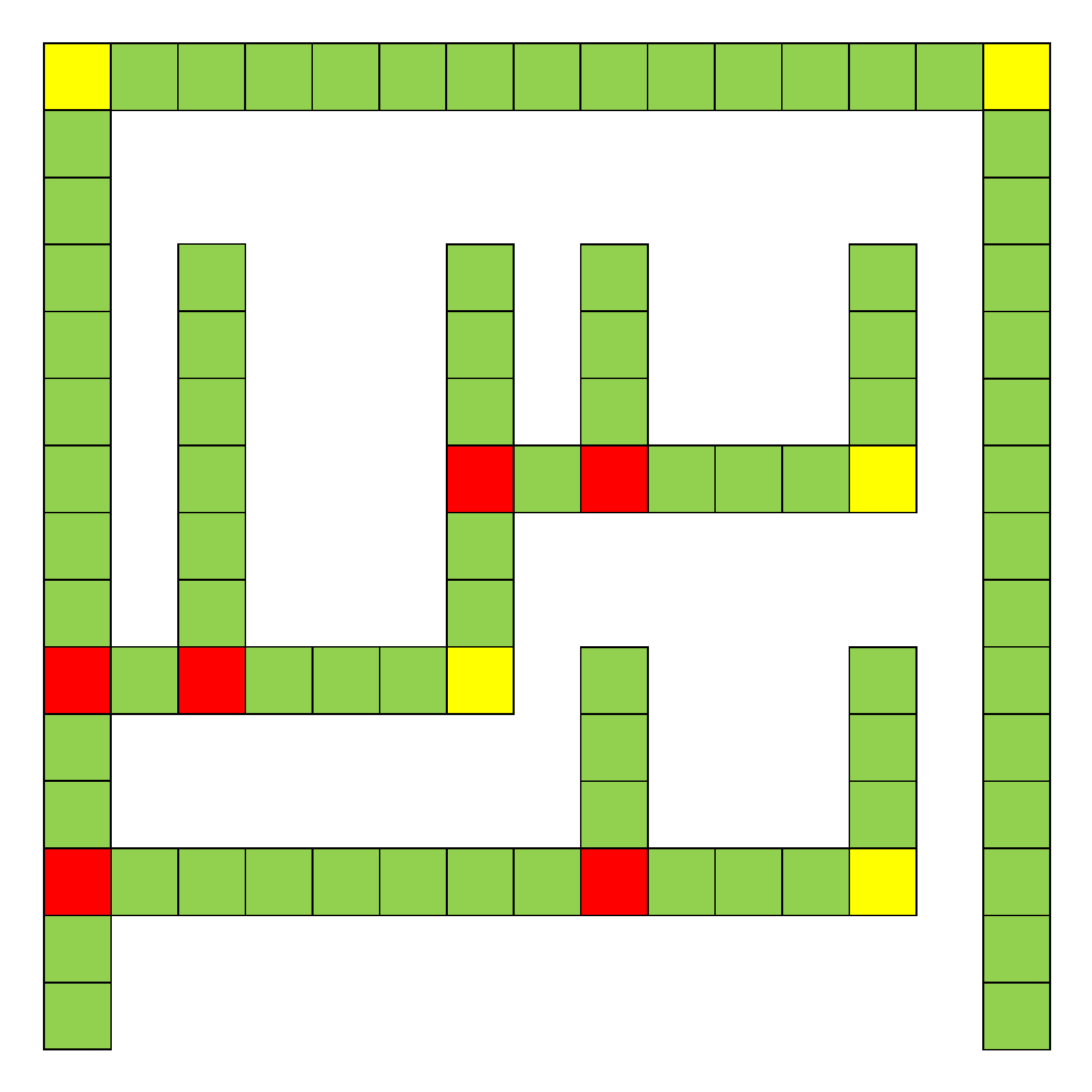}
\caption{\newnewtext{A backbone decomposed into strips (green) and corner pixels of degree two (yellow) and degree three (red).}}
\label{fig:backboneStripsandCorners}
\end{subfigure}
\hfil
\begin{subfigure}[t]{0.25\columnwidth}
	\includegraphics[width = \columnwidth]{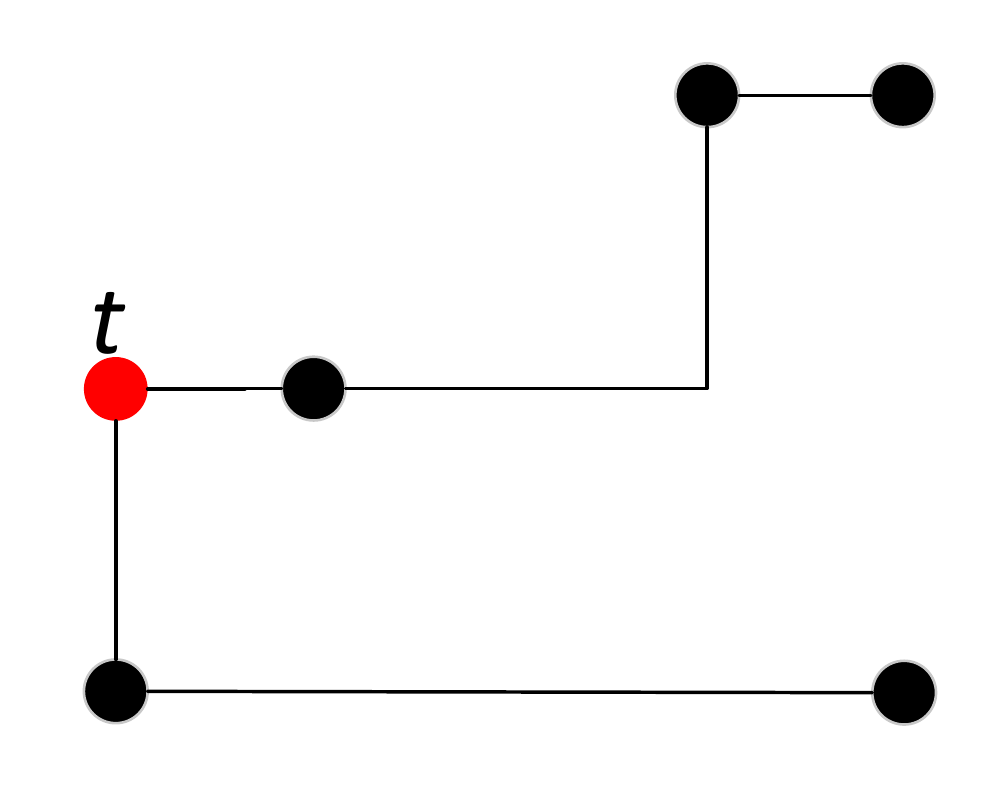}
	\caption{\newnewtext{\new{Tree $T_P$ of degree-three pixels in the backbone}. The red node is a tree median $t$ in $T_P$.}}
	\label{fig:backbonetree}
\end{subfigure}
\caption[Backbone types]{\newnewtext{Example of a backbone, its decomposition by corner pixels, the corresponding tree $T_P$ of degree-3 pixels.}}
\label{typeDec}
\end{figure}

\new{There are three levels of the decompostion, corresponding to splitting the tree by removal of {\bf (I)} degree-three
coner pixels (shown red in Figure~\ref{fig:backbonetree}), {\bf (II)} degree-two corner pixels (shown yellow in the figure) and
{\bf (III)} non-corner pixels (shown green in the figure). These splits are completely hierarchical: 
the level-I decomposition is carried out with respect to degree-three corner pixels until only pixels with degree two are 
left in all components, {as shown in Figure~\ref{typeDec}};
in level II, these components are further decomposed
with respect to the corner pixels of degree two, such that only strips
remain, i.e., subtrees without corner vertices that are shown green in Figure~\ref{typeDec}~(b).
At the third and lowest level, the straight strips are decomposed until just individual pixels are left.}

\new{There are two key properties of the decomposition: (A) bounded tree depth, which ensures a small number of stages,
and (B) bounded separation degree of subtrees, which ensures a small number of glues.
The key ingredient for (A), a polylogarithmic recursion depth for all three decomposition levels,
is that the splitting pixels are chosen such that the sizes of the split components are balanced. In
particular, we ensure that the size of components
is at most half of the size of the original component after at most two splits. This can be obtained
by choosing among the pixels of the appropriate decomposition type one that is a
(tree or path) median of a remaining backbone piece, i.e., a pixel whose removal leaves each connected component with at most half the number of vertices 
of the original tree.
Performing this splitting operation recursively yields a recursion depth of at most $\mathcal{O}(\log k)$.
In order to achieve (B), each subtree is separated from the rest of the backbone by the removal of at most two 
pixels from the rest of the backbone, special care is only necessary at level I, as subtrees in levels II and III
do not have any vertices of degree higher than two; at level I, the property is achieved by a further subdivision
into level I(a) (decomposition by removing a subtree median) and level I(b) (decomposition by removing a subpath median).
Further details are described below.}

\new{\noinbf{Level I decomposition.}}
\new{Consider a tree $T_P$ whose vertices are the pixels of degree three in the backbone;
two vertices are adjacent if their degree-three pixels can be connected by a path in the backbone
that does not contain any other degree-three pixel (see Figure \ref{fig:backbonetree}).}
Because the backbone is hole-free, $T_P$ is a tree with maximum degree three.
We decompose $T_P$ recursively as follows. Initially, we
choose a vertex $t$ (see Figure \ref{fig:backbonetree}) that splits \new{$T_P$} into connected
components that \new{each} have at most half the number of nodes, i.e., a {\em tree median} resulting
in subtrees of sizes at most $\lceil |T_P|/2 \rceil$. \new{The further decomposition 
of a nontrivial subtree $T'_P$ of $T_P$ depends on its \newest{{\em separation degree}, which is the number of degree-three pixels
that separate $T'_P$ from the rest of $T_P$.}}

\new{\noinbf{Level I(a) decomposition.}
If $T'_P$ has separation degree one, i.e., it is separated from the rest of $T_P$ by a single degree-three vertex,
we split $T'_P$ into further pieces by removing a tree median of $T'_P$; see Figure~\ref{fig:newpath}~(a). 
The resulting subtrees have separation degree one or two, and each piece has size at most $\lceil |T'_P|/2 \rceil$.}

\new{\noinbf{Level I(b) decomposition.}
If $T'_P$ has separation degree two, it is separated from the rest of $T_P$ by two degree-three vertices, say,
$v_1$ and $v_2$; see Figure~\ref{fig:newpath}~(b). 
If the path between $v_1$ and $v_2$ is a single edge in $T_P$, $T'_P$ does not contain
any further vertices, and we can proceed to level II.
If there is a nontrivial path $W$ in $T'_P$ between $v_1$ and $v_2$, we split $T'_P$ by 
picking a path median of $W$. This results in three new subtrees; two of them (say, ${T'_P}^{(1)}$ and ${T'_P}^{(2)}$) have separation degree two,
one (say, ${T'_P}^{(3)}$) has separation degree 1. Clearly, ${T'_P}^{(i)}\leq \lceil |T'_P|/2 \rceil$ for $i=1,2$. As ${T'_P}^{(3)}$ has separation 
degree 1, its next decomposition will be level I(a), ensuring that its components will have size at most $\lceil |T'_P|/2 \rceil$
after this second split.}


\begin{figure}[h!]
	\centering
	\begin{subfigure}[t]{0.42\columnwidth}
		\includegraphics[width = \columnwidth]{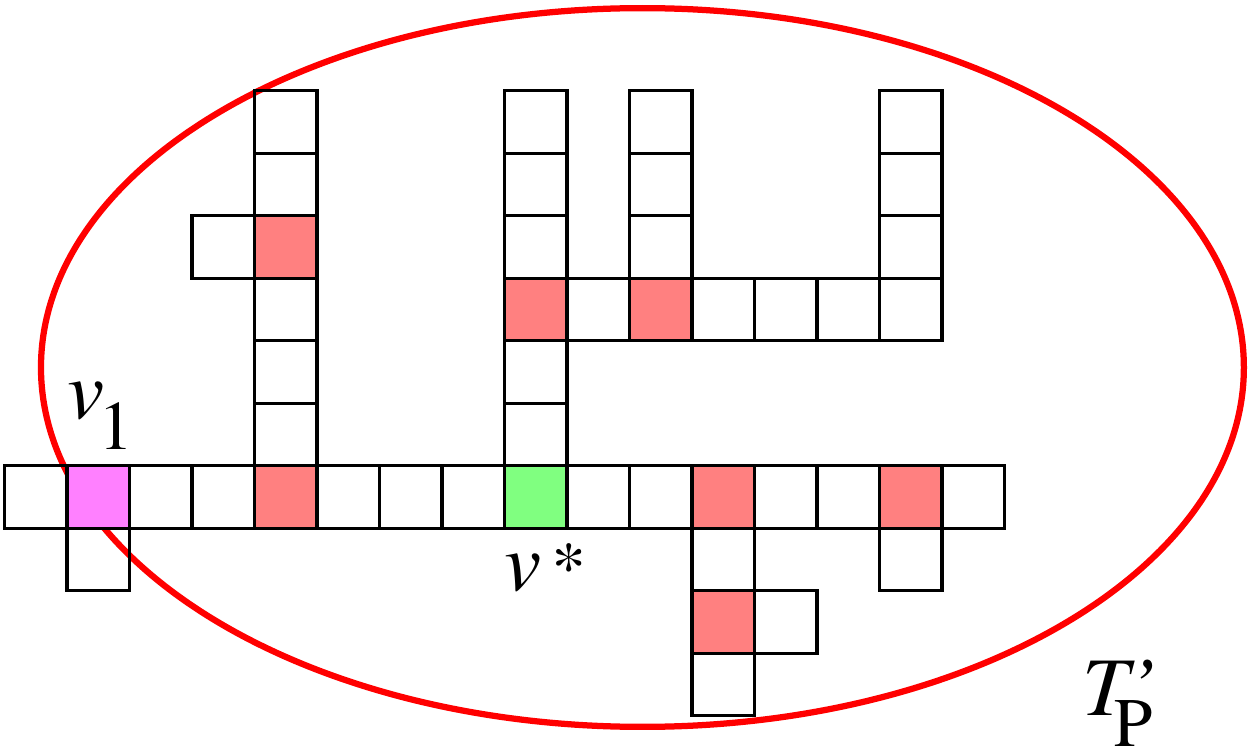}
		\caption{\newnewtext{Decomposition of level I(a): Subtree $T'_P$ has one separation vertex ($v_1$, shown in purple). The next decomposition step
is performed by splitting at the tree median, $v^*$.}}
	\end{subfigure}
	\hfil
	\begin{subfigure}[t]{0.37\columnwidth}
		\includegraphics[width = \columnwidth]{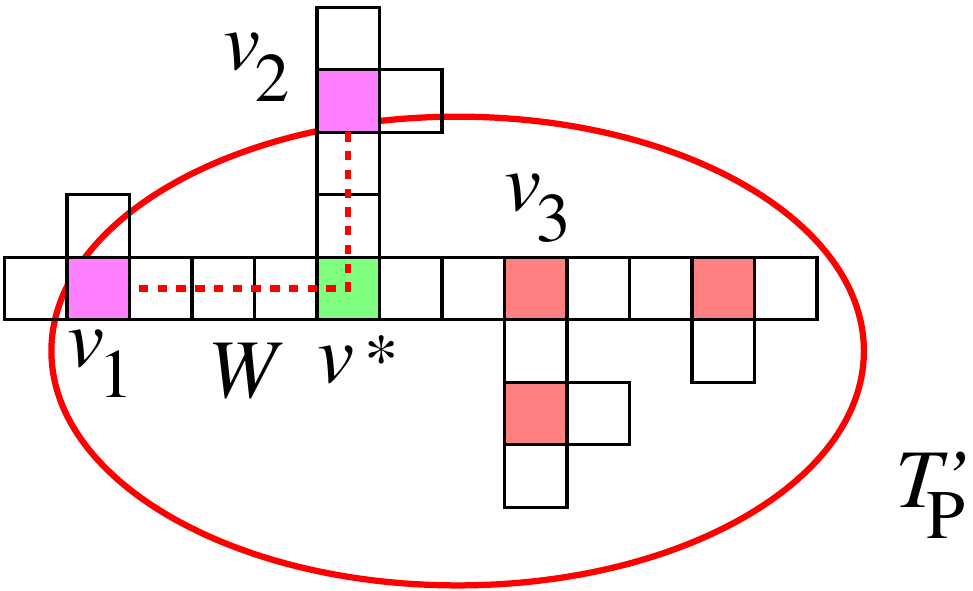}
		\caption{\newnewtext{Decomposition of level I(b): Subtree $T'_P$ has two separation vertices ($v_1$ and $v_2$, shown in purple), connected by path $W$. 
The next decomposition step is performed by splitting at the path median, $v^*$, instead of the tree median \newest{$v_3$}.}}
	\end{subfigure}
	\caption[Level I]{\newnewtext{Decomposition step of level I.}}
		\label{fig:newpath}
\end{figure}

	By induction it follows that this recursion has a depth of
$\mathcal{O}(\log |T_P|) = \mathcal{O}(\log k) \subseteq \mathcal{O}(\log n)$.
Because each node has a degree of $3$, it follows that the width of the
recursion tree is $\mathcal{O}(k)$, so a bin complexity of
$\mathcal{O}(k)$ is guaranteed for assembling all level-I components. 

\new{Assembling the subtrees of the level-I decomposition can be ensured with just three glue types, as follows;
see Figure~\ref{fig:3glues}. A degree-3 pixel $p$ is adjacent to three level-I components, say, $T_1$, $T_2$, $T_3$; 
at most two of them ($T_1$ and $T_2$) have separation degree two, so at most five connections are involved when attaching the components to $p$.
Because supertiles are not rotatable, we can separate the consideration
for horizontal connections from the one for vertical connections. At most 
four of the five connections are in the same orientation, and only if these belong to the \newest{components}
with separation degree two, as shown in Figure~\ref{fig:3glues}. 
By attaching component $T_1$ in one stage, we can also use the involved glue type $g_1$ for the 
connection of component $T_2$ that is not adjacent to $p$ in a separate, second stage, when there are no more exposed
glues of type $g_1$ at $T_1$ or $p$. 
Thus, three glue types are sufficient to ensure a unique assembly.

\begin{figure}[h!]
\centering
		\includegraphics[width = .5\columnwidth]{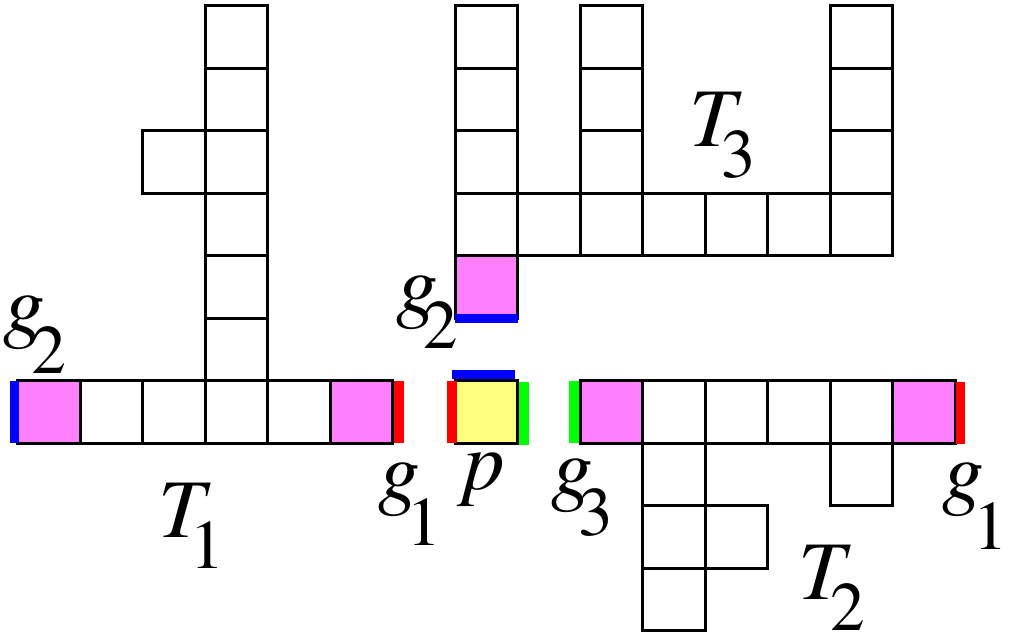}
	\caption[3glues]{\new{Assembling level-I components with three glues. In a first stage, $T_1$ gets assembled
with pixel $p$ using glue $g_1$, leaving no exposed connections with glue type $g_1$ at $p$ or $T_1$. In a second stage, $T_2$ gets attached to $p$
using glue $g_3$. Because tiles and supertiles are non-rotatable, using glue type $g_2$ for T$_3$ does not interfere with the horizontal connections.}}
		\label{fig:3glues}
\end{figure}
}

\new{\noinbf{Level II.} At the second level, we use corner pixels
of degree two for further tree decomposition.} 
We remove a pixel $t$ \newtext{of degree two} that separates $P_{W'}$ into two subpolyominoes that both have \new{almost} the same number of corner pixels of degree two. Thus, the decomposition tree has
again logarithmic height: now a node represents a corner pixel of degree $2$, while 
an edge \new{corresponding to} a straight line between two of them. By the same argument as above,
we obtain \new{a stage complexity of} $\mathcal{O}(\log k)$ and
\new{a bin complexity of} $\mathcal{O}(k)$ for all iterations of the second step.

\new{As we use a new bin, we are allowed to use the same glues as for putting together the level-I pieces.
None of these assembly steps is more complex than for level I, so 
we again conclude that \new{three} glues suffice.}


\new{\noinbf{Level III.}}
We assemble the straight lines between connection pixels of degree~$2$. For each strip, we reuse three glue types from the second step. For
every type of strip (as defined above) we take six bins to build 
\xtimes{2^\ell} strips for some $\ell \in \mathbb{N}$. To build a specific strip, we proceed as
follows. Let $m$ be the length of the strip; we build such a strip 
in one additional bin as in~\cite{DDF08}.
The individual segments for the strip are used from the bins for building strip types. Thus, we need
$\mathcal{O}(\log m)$ stages for one strip; due to parallelism, this takes $\mathcal{O}(\log n)$ stages 
for assembling all straight strips.\\

For the overall backbone assembly, we use \new{three} glues, $\mathcal{O}(1)$ tiles and $\mathcal{O}(k)$ bins 
within $\mathcal{O}(\log n\cdot\log k)$ stages: we split at tree medians $\mathcal{O}(\log k)$ times, and use $\mathcal{O}(\log n)$ stages for each strip. 
\end{proof}

\new{Now we can establish our main result for $\tau = 2$.
The main idea is to construct the backbone 
using strength 2 for each glue in the construction of Lemma~\ref{lem:backbone}
and then flooding it by a specifically designed set of tiles $S$ \new{with glues of strength 1}, which is
illustrated in Figure~\ref{gluechart}; see Figure~\ref{backboneglue} for the assembly process of flooding.  
These additional tiles are attached in a cooperative
manner, i.e., by making use of two strength-1 glues for each attachment. 
Similar to the construction of Theorem~\ref{squareTheorem},
this ensures that precisely the pixels of the original polyomino are filled in.}


\begin{theorem}
\label{th:hole2}
Let $P$ be an arbitrary polyomino with $k$ vertices. Then there is a $\tau=2$
staged assembly system that constructs a fully connected version of $P$
in $\mathcal{O}(\log^2 n)$ stages, with \new{$6$} glues, $\mathcal{O}(1)$ tiles,
$\mathcal{O}(k)$ bins and scale factor $3$.  
\end{theorem}

\begin{figure}[h!]
\centering
\includegraphics[width = 0.25\columnwidth]{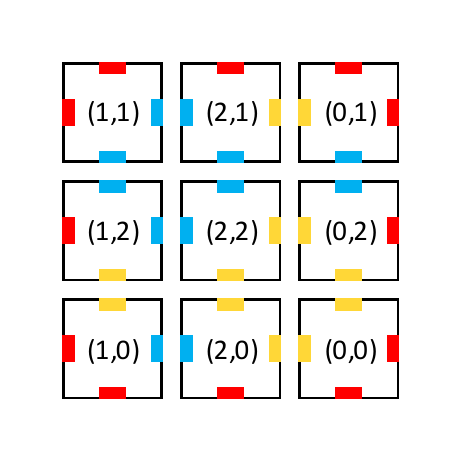}
\caption[$3\times 3$ glue chart]{Glue chart for \xtimes[3]{3} tiles for filling
up the shape. \new{The set $S$ of {\em flooding tiles} consists of the nine tiles shown in the figure}. Blue glue $\stackrel{\wedge}{=} g_4$, yellow glue
$\stackrel{\wedge}{=} g_5$ and red glue $\stackrel{\wedge}{=} g_6$.}
\label{gluechart}
\end{figure} 

\begin{figure}[h!] 
\centering
\begin{minipage}{0.3\columnwidth}
\includegraphics[width = \columnwidth]{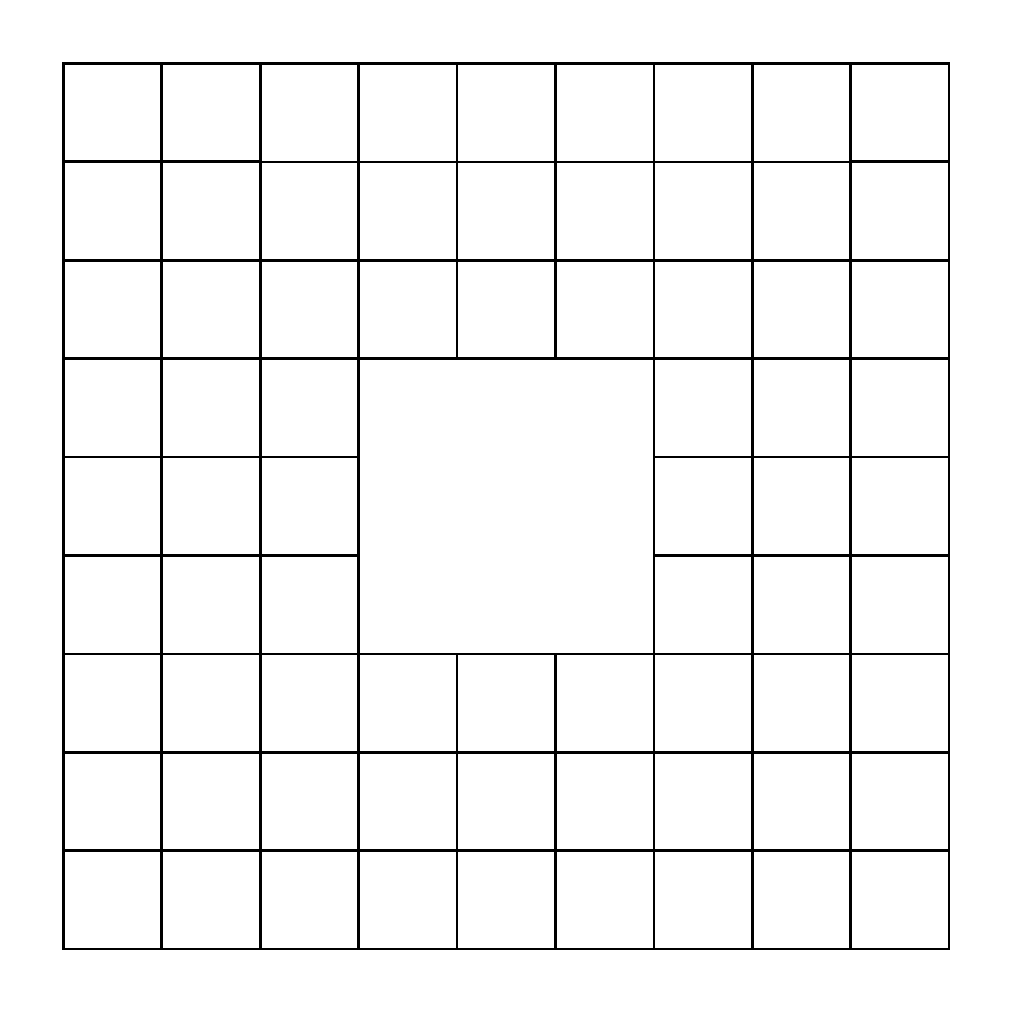}
\label{backboneAA}
\end{minipage}
\hfil
\begin{minipage}{0.36\columnwidth}
\includegraphics[width = \columnwidth]{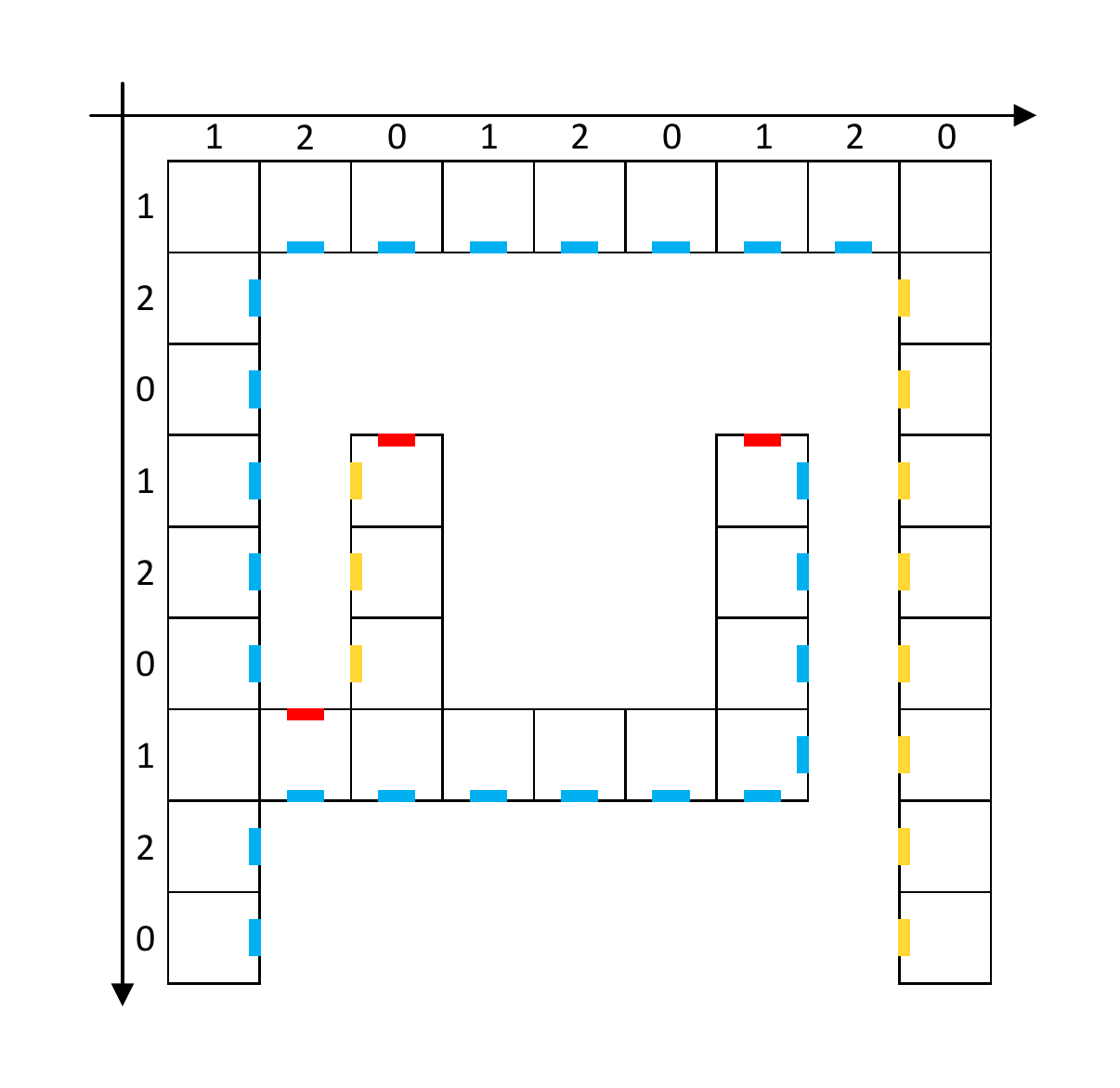}
\label{backboneD}
\end{minipage}
\hfil
\begin{minipage}{0.27\columnwidth}
\includegraphics[width = \columnwidth]{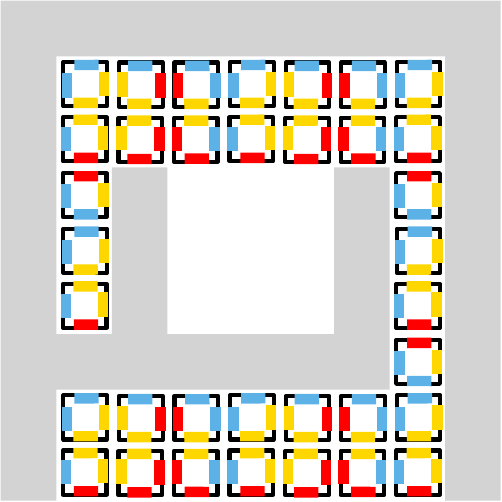}
\label{backboneCC}
\end{minipage}
\caption[Glues of a backbone]{(Left) \new{A polyomino $P^3$, obtained by scaling a polyomino $P$ with one hole by a factor of 3. (Middle)
The backbone of $P^3$ with strength-1 glues assigned to the inside edges; note the coordinates modulo 3
and their correspondence to the glue chart.
(Right) Filling up the backbone (shaded in grey) with the flooding tiles to assemble the polyomino.}}
\label{backboneglue}
\end{figure}

\begin{proof}
	\new{We assemble a polyomino $P^3$, obtained by scaling a polyomino $P$
by a factor of 3. Because of the scaling, all relevant coordinates of $P^3$ are multiples of 3.
After assembling the backbone $B(P^3)$ of $P^3$ (according to Lemma~\ref{lem:backbone}, making use of strength-2 glues $g_1$, $g_2$ and $g_3$),
we consider the {\em inside edges} of $B(P^3)$, which are the edges of pixels of $B(P^3)$
that are incident to pixels in $P^3\setminus B(P^3)$.
Now the following is straightforward to verify from the construction of the backbone; refer to Figure~\ref{backboneglue}~(Middle).

\begin{itemize} 
\item All inside edges facing east separate pixels at $x$-coordinates 1 and 2 modulo 3, with 
the backbone pixel at 1 modulo 3.
\item All inside edges facing south separate pixels at $y$-coordinates 1 and 2 modulo 3, with 
the backbone pixel at 1 modulo 3.
\item All inside edges facing west separate pixels at $x$-coordinates 2 and 0 modulo 3, with 
the backbone pixel at 0 modulo 3.
\item All inside edges facing north separate pixels at $y$-coordinates 0 and 1 modulo 3, with 
the backbone pixel at 1 modulo 3.
\end{itemize}

Now we specify the set of
flooding tiles $S$; the idea is similar to the flooding tiles of Rothemund
and Winfree~\cite{Rot00} described in \newest{our} Theorem~\ref{th:one}.  The mechanism is illustrated in
Figure~\ref{gluechart} and Figure~\ref{backboneglue}. }

	\new{$S$ consists of the set of nine tile types, shown in
Figure~\ref{gluechart}. In addition to the three strength-2 glues used for assembling the backbone,
we use three additional strength-1 glue types ($g_4$, $g_5$, $g_6$) that are assigned to the edges
of tiles from $S$. Note the assignment modulo 3, corresponding to $x$- and $y$-coordinates as follows. 

\begin{itemize}
\item A pixel of type $(1,1)$ gets glue type $g_4$ (east), $g_4$ (south), $g_6$ (west), $g_6$ (north).
\item A pixel of type $(1,2)$ gets glue type $g_4$ (east), $g_5$ (south), $g_6$ (west), $g_4$ (north).
\item A pixel of type $(1,0)$ gets glue type $g_4$ (east), $g_6$ (south), $g_6$ (west), $g_5$ (north).
\item A pixel of type $(2,1)$ gets glue type $g_5$ (east), $g_4$ (south), $g_4$ (west), $g_6$ (north).
\item A pixel of type $(2,2)$ gets glue type $g_5$ (east), $g_5$ (south), $g_4$ (west), $g_4$ (north).
\item A pixel of type $(2,0)$ gets glue type $g_5$ (east), $g_6$ (south), $g_4$ (west), $g_5$ (north).
\item A pixel of type $(0,1)$ gets glue type $g_6$ (east), $g_4$ (south), $g_5$ (west), $g_6$ (north).
\item A pixel of type $(0,2)$ gets glue type $g_6$ (east), $g_5$ (south), $g_5$ (west), $g_4$ (north).
\item A pixel of type $(0,0)$ gets glue type $g_6$ (east), $g_6$ (south), $g_5$ (west), $g_5$ (north).
\end{itemize}

Furthermore, pixels in the backbone get glues assigned to their inside edges, as follows.

\begin{itemize} 
\item All inside edges facing east or south get strength-1 glue $g_4$.
\item All inside edges facing west get strength-1 glue $g_5$.
\item All inside edges facing north get strength-1 glue $g_6$.
\end{itemize}

Note that this assignment of $g_4$, $g_5$, $g_6$ to backbone pixels already happens when assembling the backbone, as described
in Lemma~\ref{lem:backbone}; because these glues only face east/south, west, and north, respectively, and 
tiles and supertiles cannot be rotated, no bonding with these glues is possible between backbone pixels.

For filling the polyomino, we mix the nine kinds of flooding pixels (as shown in Figure \ref{gluechart}) with the backbone supertiles in one bin. 
Now it is straightfoward to verify by induction that every pixel in $P^3\setminus B(P^3)$ will get attached
to the backbone by a cooperative sequence of assembly steps that each use two strenth-1 bonds; this is
completely analogous to the technique described in \newest{our} Theorem~\ref{th:one}, with glues $g_4$, $g_5$ and $g_6$ being
compatible at each step because of their coordinates modulo 3. Also analogous to the construction of
squares in \newest{our} Theorem~\ref{th:one}, the converse holds: A pixel $p=(p_x,p_y)\not\in B(P^3)$ can 
only be attached in this manner if first
backbone pixels (say, $q$ and $r$) encountered when traveling from $p$ in two axis-parallel direction are met at inside edges,
i.e., only when $p$ belongs to $P^3$.
(See Figure~\ref{fig:monotone}.) This property only holds for pixels in $P^3$, implying
that precisely $P^3$ gets assembled.}

\begin{figure}[h!]
\centering
\includegraphics[width = 0.25\columnwidth]{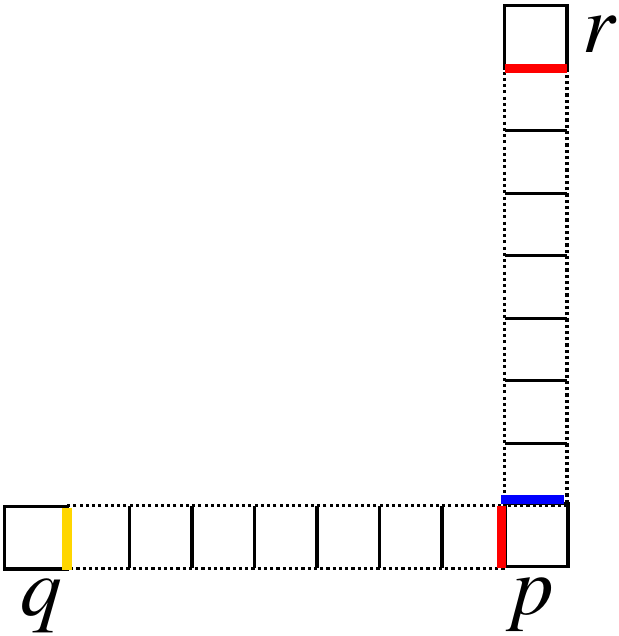}
\caption[Not in pixels]{\new{Attaching a pixel $p\not\in B(P^3)$ to the backbone by a sequence of flooding tiles is only possible if the first 
backbone pixels encountered when traveling from $p$ in two axis-parallel direction are met at inside edges,
i.e., only when $p$ belongs to $P^3$.}}
\label{fig:monotone}
\end{figure} 

\new{Overall, we have three glue types for building the backbone and three glue types for
the flooding tiles for building the interior of the polyomino. Hence, we use a total of six glue types and $\mathcal{O}(1)$ tile types.}

In total, we need $\mathcal{O}(\log^2 n)$ stages, \new{six} glues, $\mathcal{O}(1)$
tiles and $\mathcal{O}(k)$ bins to assemble a fully connected polyomino, scaled
by a factor 3 from the target shape. 
\end{proof}

As noted before, the number of degree-3 corner pixels depends on the number of holes. 
We can describe the overall complexity in terms of $h$, the number of holes.
For the special case of hole-free shapes, we can skip some steps, reducing
the necessary number of stages. {In particular, Corollary~\ref{cor:stage} follows from Theorem~\ref{th:hole2}.}

\begin{corollary}\label{cor:stage}
The stage complexity of Theorem \ref{th:hole2} can be quantified in the the number of holes $h$ such we get a stage complexity of \new{$\mathcal{O}(\log^2\, h + \log\, n)$}.
In particular, Theorem \ref{th:hole2} gives a staged self-assembly system for
hole-free shapes with \new{$\mathcal{O}(\log\, n)$} stages, \new{six} glues,
\new{$\mathcal{O}(1)$} tiles, \new{$\mathcal{O}(k)$} bins and a scale factor of
3.  
\end{corollary}

\section{Fully Connected Constructions for $\tau=1$}\label{t=1}
In this section we describe approaches for assembling \newnewtext{fully connected} polyominoes at temperature $\tau=1$. 

\subsection{Hole-Free \new{Polyominoes}, $\tau=1$}\label{holefreeshape}
We present a system for building hole-free 
polyominoes. The main idea is based on \cite{DDF08}, i.e., splitting the polyomino 
into strips. Each of these strips gets assembled piece by piece; if there
is a component that can attach to the current strip, we create it
and attach it.

Our geometric approach partitions the polyomino into rectangles and
uses them to assemble the whole polyomino. Even for complicated shapes with many vertices,
this number of rectangles is never worse than quadratic in the size of the bounding box;
in any case we get a large improvement \new{of} the stage complexity.  

We first consider a building block, see Figure~\ref{DegSquare}. \new{Originally a rectangle, its shape gets modified by
{\em tabs} and {\em pockets} that fit together like key and lock. This is
based on the jigsaw technique of~\cite{DDF08}, whose idea is to use the geometric shape for ensuring unique assembly.}
In our case tabs are \new{$1\times 1$ or $1\times 2$ rectangles that are
attached to the considered rectangle} and pockets are \new{$1\times 1$ or $1
\times 2$ rectangles that are} missing \new{from} the \new{boundary of the}
considered rectangle.

\begin{lemma}
\label{degsquare}
A \xtimes[2n]{2m} rectangle (with $n\geq m$) with at most two tabs at top and left side and at
most two pockets at each bottom or right side (see Figure \ref{DegSquare}) can be
assembled with $\mathcal{O}(\log n)$ stages, $9$ glues, $\mathcal{O}(1)$ tiles
and $\mathcal{O}(1)$ bins at $\tau=1$.  \end{lemma}

\begin{figure}[h!]
\centering
\includegraphics[scale = 0.5]{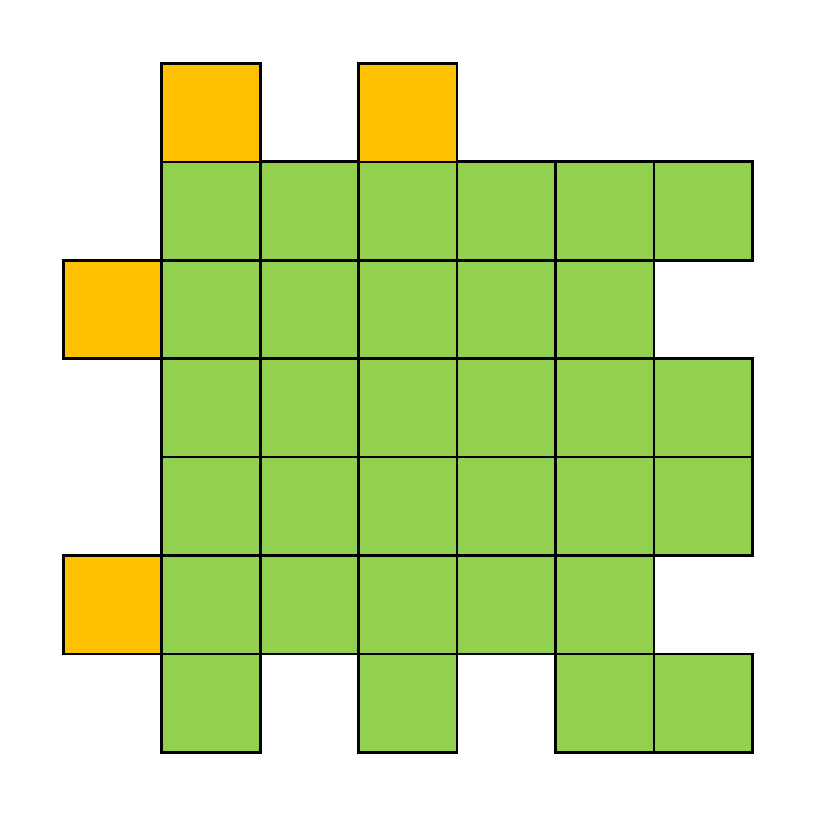}
\caption[Square-like shape with tabs and pockets]{A \new{rectangular} shape (green) with tabs on top and left side (orange), and pockets on bottom and right side.}
\label{DegSquare}
\end{figure}

\begin{figure}[h!] 
\centering
\begin{subfigure}[t]{0.20\columnwidth}
	\includegraphics[width = \columnwidth]{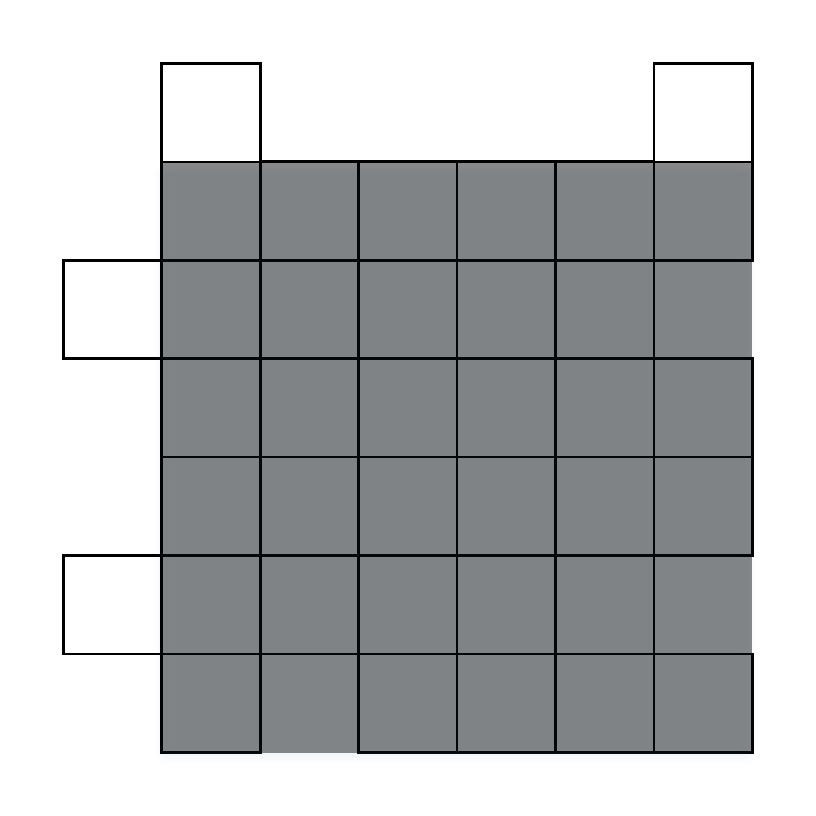}
	\caption{{\new{Shaded in black is the original rectangle without tabs and pockets.}}}
	\label{fig:nmSquare}
\end{subfigure}
\hfil
\begin{subfigure}[t]{0.24\columnwidth}
	\includegraphics[width = \columnwidth]{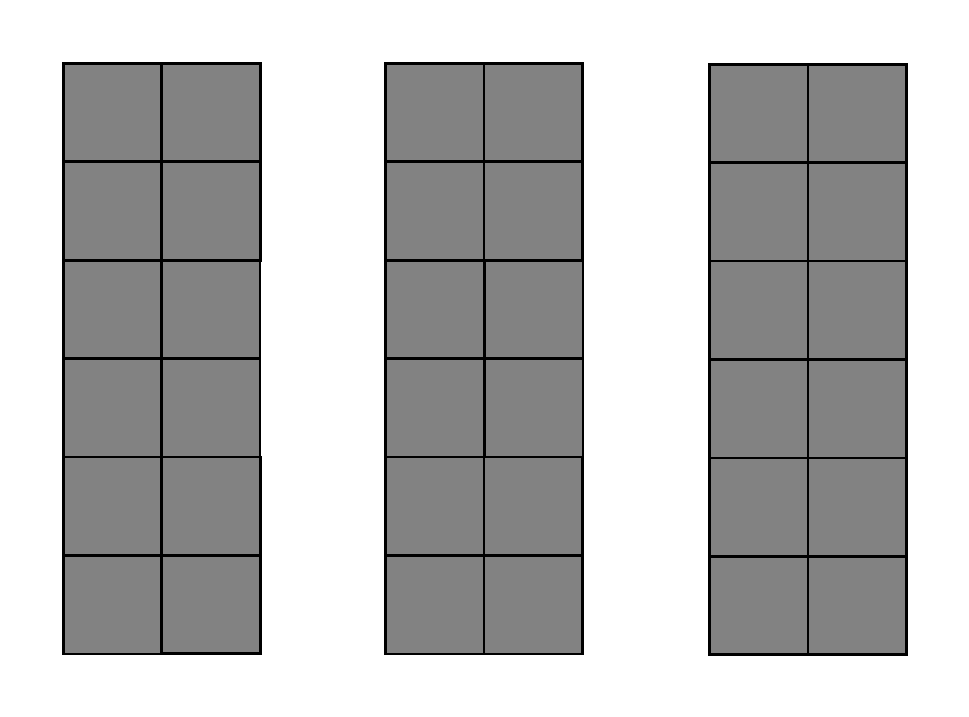}
	\caption{\new{Partition of the original rectangle into width-2 rectangles}}
	\label{fig:width2Rects}
\end{subfigure}
\hfil
\begin{subfigure}[t]{0.2\columnwidth}
	\includegraphics[width = \columnwidth]{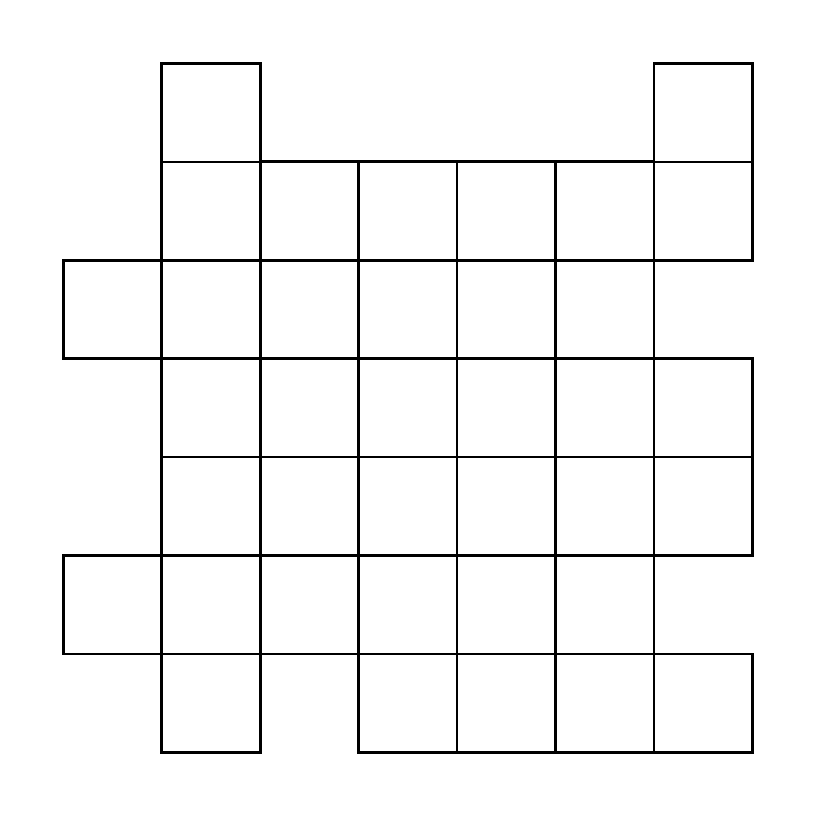}
	\caption{\new{The modified rectangle with tabs and pockets.}}
	\label{fig:modified}
\end{subfigure}

\begin{subfigure}[h!t]{0.28\columnwidth}
\includegraphics[width = \columnwidth]{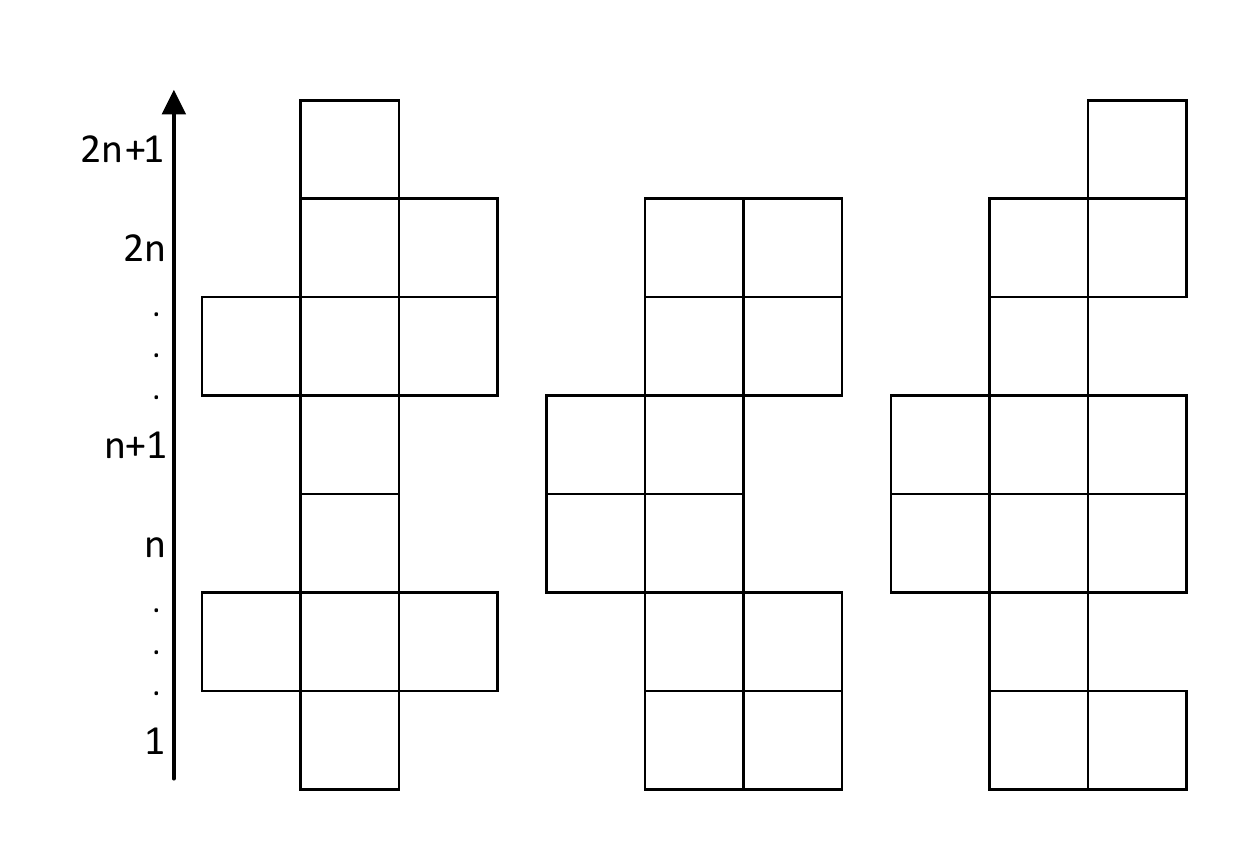}
\caption{\new{Adding bulges, tabs and pockets to the width-2 rectangles. Note the coordinate labels for the rows.}}
\label{fig:holefreeb}
\end{subfigure}
\hfil
\begin{subfigure}[t]{0.28\columnwidth}
	\includegraphics[width = \columnwidth]{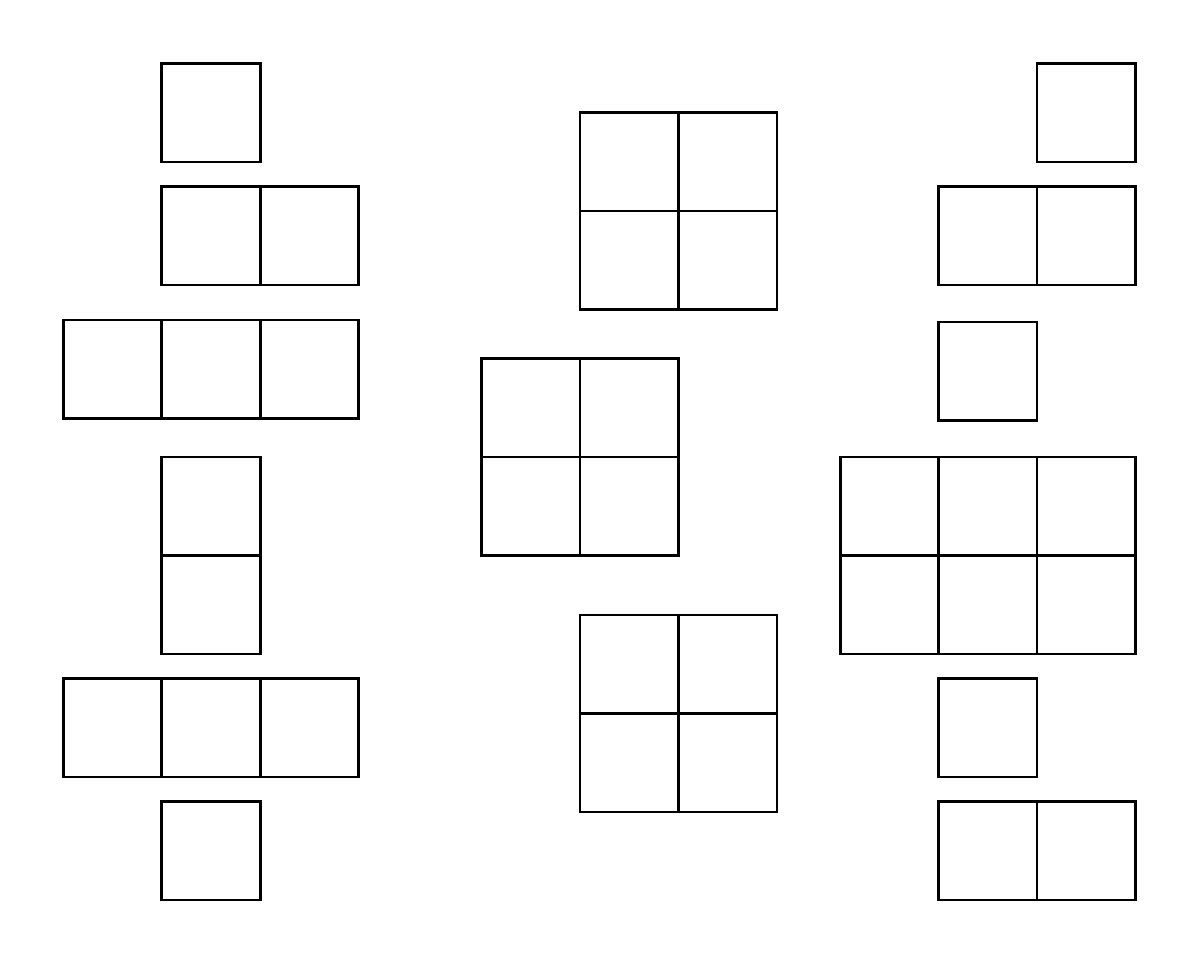}
	\caption{\new{Cutting the components horizontally. The left and the right component have tabs and pockets, the middle component has none.}}
	\label{fig:holefree}
\end{subfigure}
\caption[Decomposition of a modified square]{\newnewtext{Stepwise decomposition
of a modified rectangle. \new{The union of the dark shaded pixels 
indicates the initial rectangle whose shape is extended and reduced by tabs and
pockets}.}} \label{partsquare}
\end{figure}

\begin{proof}{
\new{Refer to Figure~\ref{partsquare} for the overall construction.}
First consider the \xtimes[2n]{2m} rectangle without any tabs or pockets (\new{shaded dark} in Figure~\ref{fig:nmSquare}), which we partition into
(vertical) rectangles of width 2 (see Figure~\ref{fig:width2Rects}). \new{Analogously, the modified rectangle shown in Figure~\ref{fig:modified}
gets dissected into {\em components}, i.e., pieces that are}
joined by \newtext{bulges} in rows $n$ and $n+1$, \new{in addition to} the tabs
and pockets from the shape we need to build (see Figure~\ref{fig:holefreeb}).
\new{For assembly, we use glues on their sides like they \new{are} used in the jigsaw technique of~\cite{DDF08}.}
Now every component has a maximum width of 3, even with the tabs. 
This allows us to use nine glues to create each component with
attached tabs and pockets as follows. 

A component without a tab or pocket 
\new{(e.g., the middle component in Figure~\ref{fig:holefree})}
is cut between the 
$(n-1)$st and the $n$th row, as well as between the $(n+1)$st and the $(n+2)$nd row.
Then we have two strips of width 2 and one \xtimes[2]{2} square. The square can
be assembled by brute force with desired glues on its sides. The strips can also be
decomposed recursively like a \xtimes{n} strip with desired glues on the
sides. Thus, for this kind of component, nine glues suffice and the component is
built within $\mathcal{O}(\log n)$ stages. Note that we need $\mathcal{O}(1)$
bins to store every possible component of this kind, i.e., they use
one out of three possible glue triples on each side \newnewtext{(compare \new{to} square assembly in \cite{DDF08}).}
\newnewtext{Because the left side of a component uses completely different glue types than the right side, a component will never attach to itself.}

A component with tabs and/or pockets  
\new{(e.g., the left or right component in Figure~\ref{fig:holefree})}
is cut between rows, such
that only components without tabs and pockets and at most four components with
tabs and pockets exists. Note that the four components are either single tiles
or \new{lines of length at most three}. The other components are either strips of width two or
similar to the components above. Hence, we need at most $\mathcal{O}(\log n)$
stages to build the biggest component. Then we assemble all components by successively
putting together pairs. We observe that this kind of component appears at most six times.
Thus, we need six  bins to store the components of this kind. Again the nine glues
suffice.

Now we have all components in $\mathcal{O}(1)$ bins, so we can 
assemble the components in a pairwise fashion to the desired polyomino within $\mathcal{O}(\log n)$
stages. Overall, nine glues suffice, so we have $\mathcal{O}(1)$ tiles.
}
\end{proof}

\begin{theorem}
\label{th:nohole1}
Let $P$ be a hole-free polyomino with $k$ vertices. Then there is a $\tau=1$
staged assembly system that constructs a fully connected version of $P$
in $\mathcal{O}(\log^2 n)$ stages, with $18$ glues, $\mathcal{O}(1)$ tiles,
$\mathcal{O}(k)$ bins and scale factor $4$.  
\end{theorem}

\begin{proof}
We cut the polyomino $P$ with horizontal lines, such that all cuts go through reflex
vertices of $P$, leaving a set of rectangles. 
If $V_r$ is
the set of reflex vertices \new{of} the polyomino, we have at most $|V_r|=:k_r$ cuts and
therefore $\mathcal{O}(k)$ rectangles.
 \newtext{Consider the rectangle adjacency tree, i.e., a graph whose vertices are the rectangles and an edge connects two vertices if the corresponding rectangles have a side in common.} \newnewtext{Note that the resulting adjacency graph is a tree because $P$ has no holes.} 
\new{As shown in Figure~\ref{fig:holea}, we find a} rectangle \new{$R$}
that forms a tree median in the rectangle adjacency tree, i.e., a rectangle that splits the tree into connected
components that \new{each} have at most half the number of rectangles. 

Recursing over
this splitting operation \new{yields} a tree decomposition of depth 
$\mathcal{O}(\log k)$. On the pieces,
we use a scale factor of 2 for employing a jigsaw decomposition \newtext{corresponding to \cite{DDF08}.}

\begin{figure}[h!] 
\centering
	\begin{subfigure}{0.45\textwidth}
	\includegraphics[width = \columnwidth]{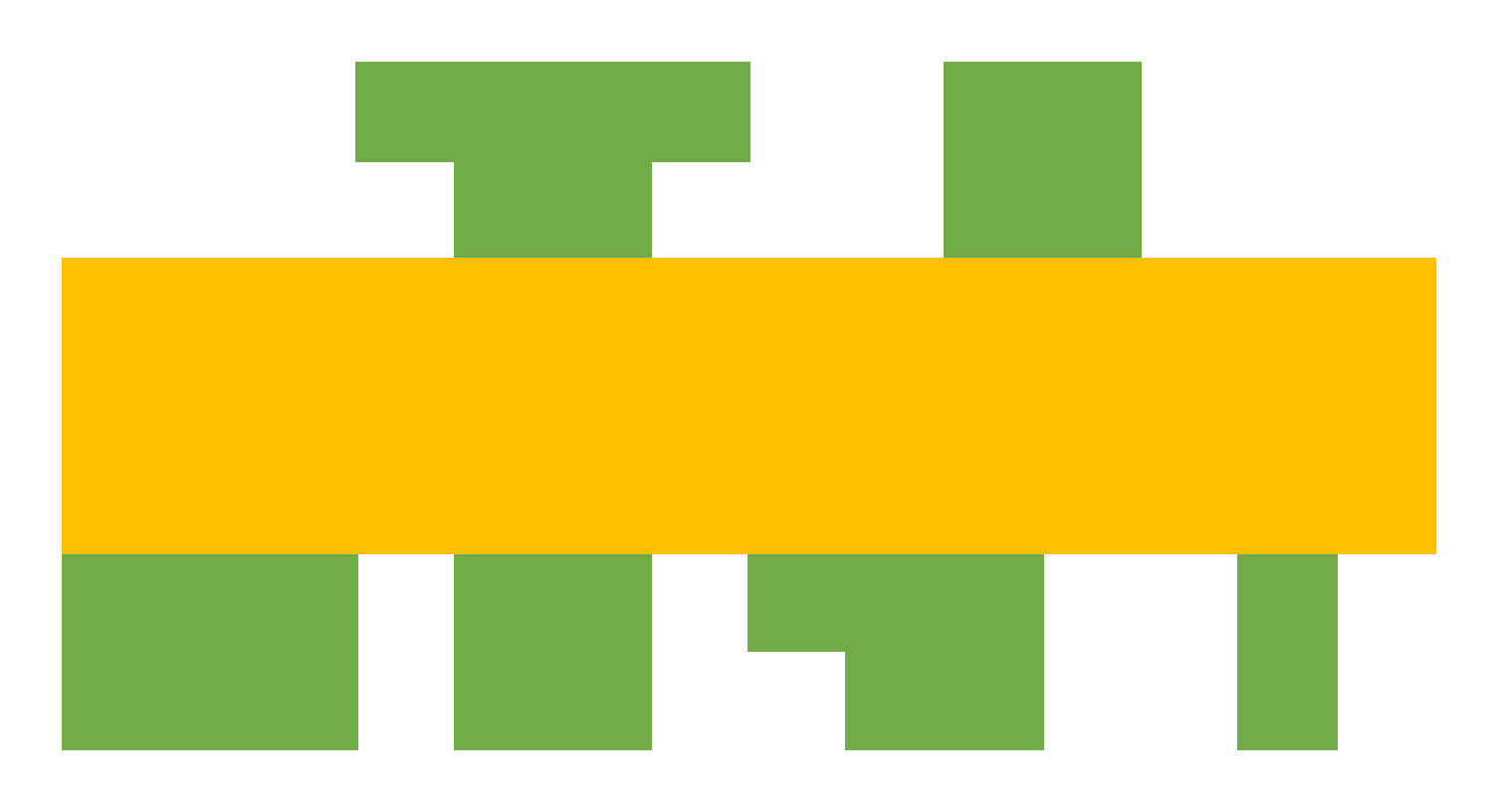}
	\caption{\new{A splitting rectangle (orange) whose removal partitions the polyomino into components (green).}}
	\label{fig:holea}
	\end{subfigure}
\ \\
\hfil
	\begin{subfigure}{0.35\textwidth}
	\includegraphics[width = \columnwidth]{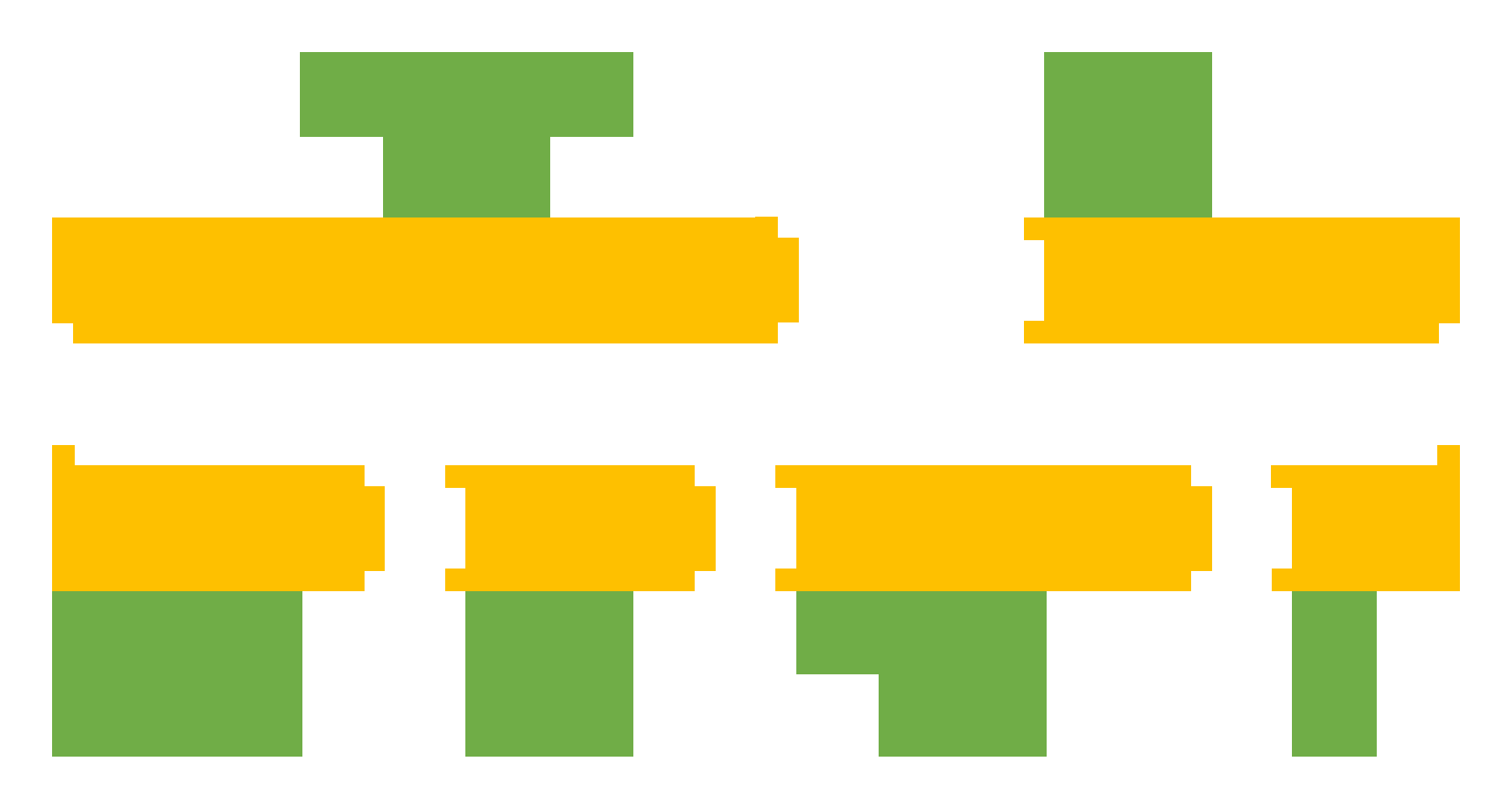}
	\caption{\new{Decomposition of the splitting rectangle.}}
	\label{fig:holeb}
	\end{subfigure}
\hfil
	\begin{subfigure}{0.35\textwidth}
	\includegraphics[width = \columnwidth]{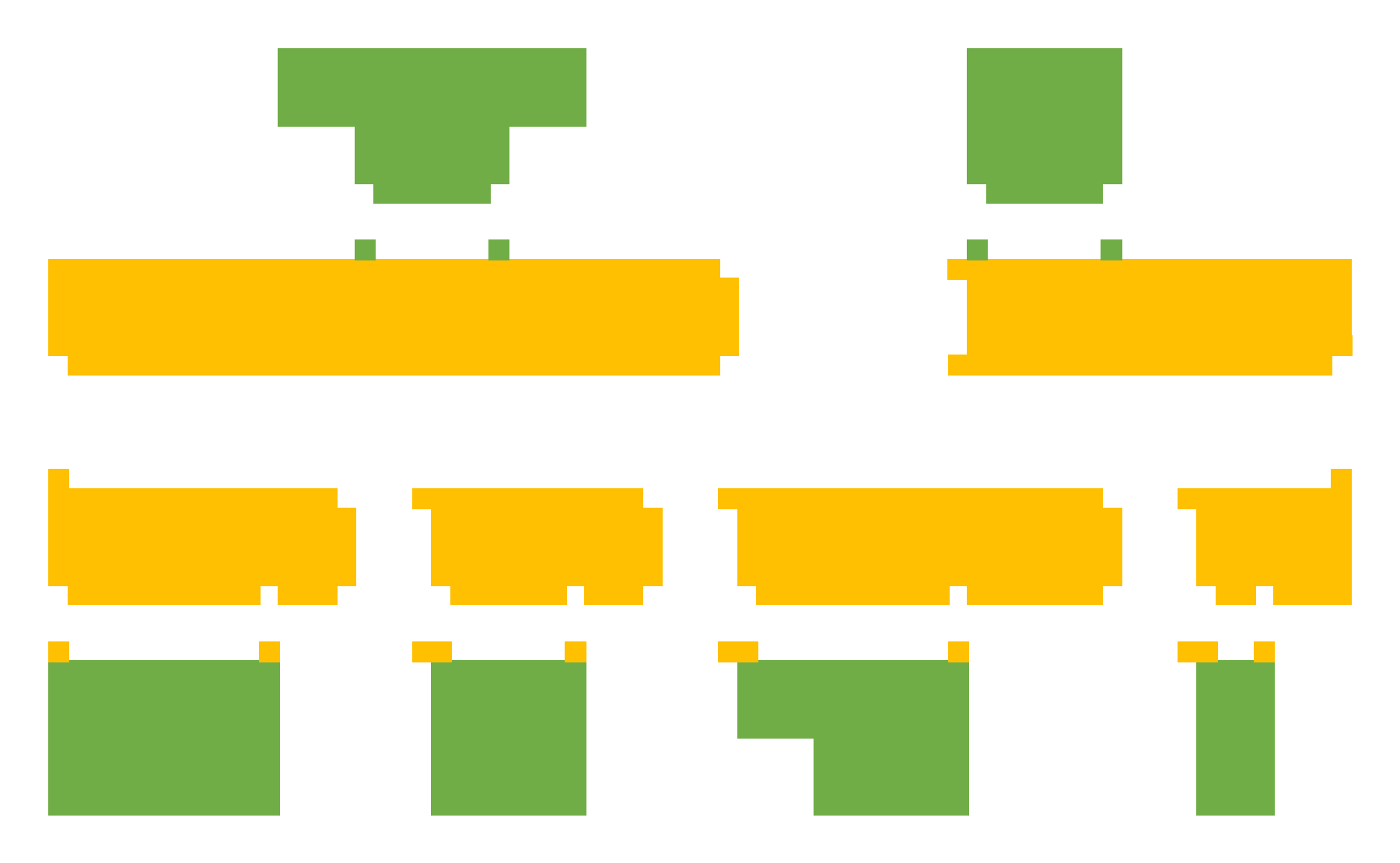}
	\caption{\new{Further decomposition of the resulting connected components.}}
	\label{fig:holec}
	\end{subfigure}
\caption[Decomposition of hole-free shapes]{\new{Splitting operation by a median rectangle and further disassembly.}}
\label{holefree}
\end{figure}

\new{When considering a single split, the polyomino without $R$ decomposes into
a number of connected components, corresponding to the green pieces in
Figure~\ref{fig:holea}}. To assemble all of \new{these components into the
original polyomino}, we employ another scale factor of 2, allowing us to split
$R$ in
half with a horizontal line, \new{as shown in Figure~\ref{fig:holeb}. Each half is further subdivided vertically into 
jigsaw components, such that each component can connect to a
part of $R$ independently from the others, \new{as shown in Figure~\ref{fig:holec}}.}
\newtext{For our purpose, we place the cuts for this subdivision such that they run along the leftmost side where a component of the chosen rectangle and its adjacent rectangle meet.}
When all components have been attached to some part of $R$, we can assemble
both halves of $R$ and then put these two together. Doing this for all
rectangles produces $\mathcal{O}(k_r)$ new components. Hence, our
decomposition tree has at most $\mathcal{O}(k_r) = \mathcal{O}(k)$ leaves, where
the leafs are rectangles that need $\mathcal{O}(\log n)$ \newtext{stages} for construction.
This yields $\mathcal{O}(\log k \log n)$ stages overall; the rectangle components
consume $\mathcal{O}(k)$ bins. Similar to
assembling a square, we need nine glues to uniquely assemble all rectangles to the
correct polyomino.

By construction, every rectangle component\newtext{, i.e., a leaf of the decomposition tree,} has at most four adjacent rectangle components, \newtext{because we decompose a chosen rectangle until all of the rectangle components have at most four adjacent rectangles};
its size is \xtimes[2w]{2h} for some width $w$ and height $h$. The four
adjacent components are all connected at different sides, so the left and
upper side each have two tabs, while the right and lower side have two pockets.
Thus, we can use the approach of Theorem~\ref{degsquare} to
assemble all rectangles with $9$ additional glues and $\mathcal{O}(1)$ bins for
each rectangle component.

Overall, we have $\mathcal{O}(\log n)$ stages to assemble the $\mathcal{O}(k)$
rectangles with $\mathcal{O}(1)$ bins for each rectangle, plus
$\mathcal{O}(\log^2 n)$ stages to assemble the polyomino from the rectangles,
for a total of $\mathcal{O}(\log^2 n)$ stages and $\mathcal{O}(k)$ bins.
For the rectangles we need nine glues, along with nine glues for the remaining assembly,
for a total of 18 glues, with $\mathcal{O}(1)$ tile types. The overall scale factor
is 4. 
\end{proof}

\newest{With respect to later application in Theorem~\ref{th:hole1}, we remark that the construction of Theorem~\ref{th:nohole1}
hinges on sufficient vertical thickness of the constructed polyomino; the scale factor is only used to guarantee this 
thickness.

\begin{corollary}
\label{cor:thick}
Let $Q$ be a hole-free polyomino with $k$ vertices, such that $Q$ has vertical thickness at least 4, i.e.,
every maximal connected set of pixels in $Q$ with the same $x$-coordinate 
contains at least four elements. Then there is a $\tau=1$
staged assembly system that constructs a fully connected version of $Q$
in $\mathcal{O}(\log^2 n)$ stages, with $18$ glues, $\mathcal{O}(1)$ tiles,
$\mathcal{O}(k)$ bins. 
\end{corollary}

The construction is identical to the one of Theorem~\ref{th:nohole1}; note that no scaling is necessary,
as the vertical thickness suffices to allow the required horizontal splitting.
}

\subsection{\boldmath \newtext{Polyomino} with Holes, $\tau = 1$}\label{holeshape1}

In this section we \new{give staged assembly systems with temperature $\tau=1$ for arbitrary polyominoes that may have holes.}

\begin{theorem}
\label{th:hole1}
Let $P$ be an arbitrary polyomino with $k$ vertices. Then there is a $\tau=1$
staged assembly system that constructs 
\new{in $\mathcal{O}(\log^2 n)$ stages, with $20$ glues, $\mathcal{O}(1)$ tiles,
$\mathcal{O}(n)$ bins a fully connected supertile $P^6$ that arises from $P$
by a scale factor of $6$.}
\end{theorem}

\begin{figure}[h]
    \centering
	\begin{subfigure}[t]{\columnwidth}
    		\centering
      		\includegraphics[height=7.5cm]{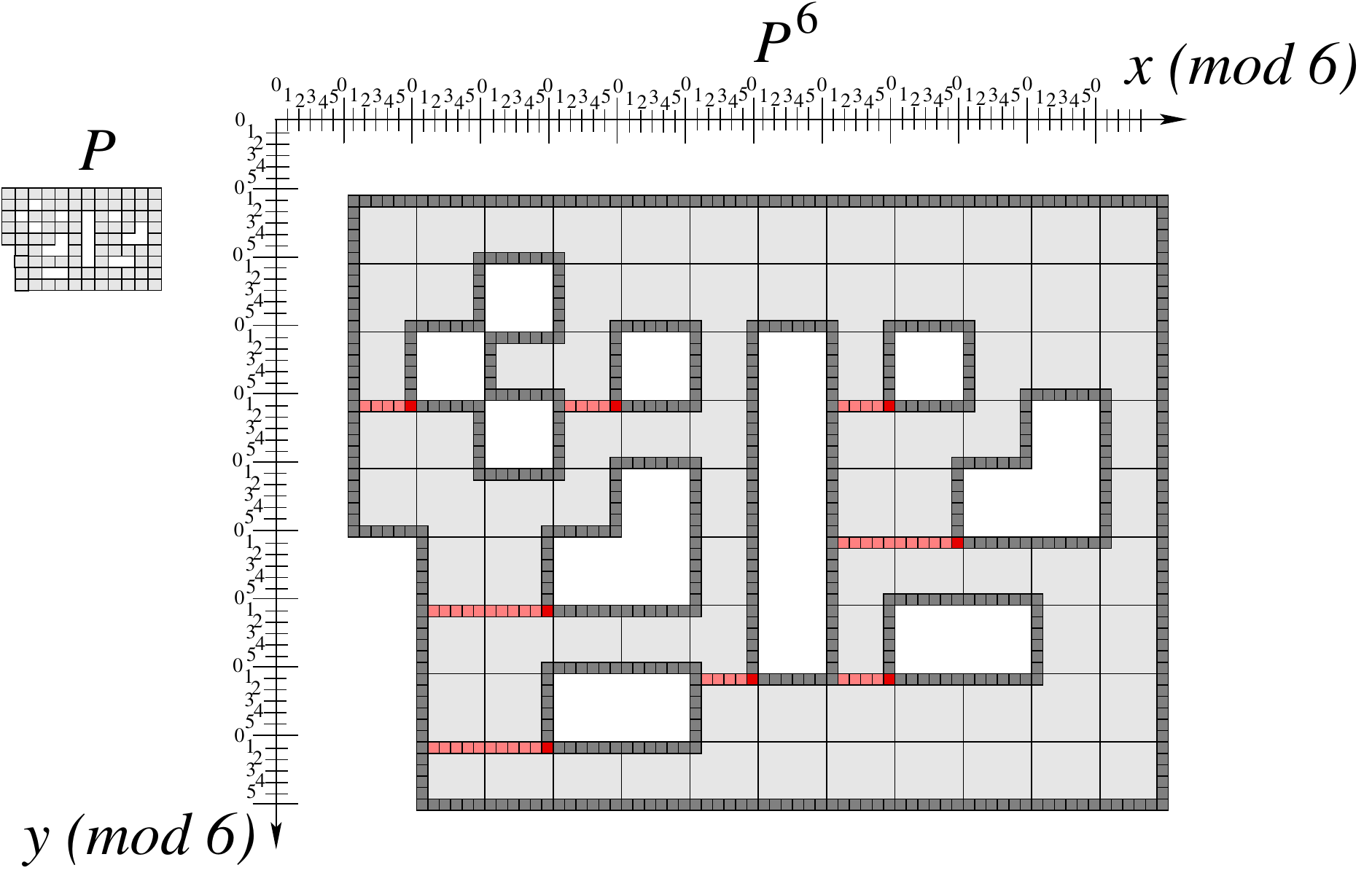} 
    		\caption{\new{An unscaled polymino $P$, the scaled $P^6$ and the construction of $S_1$. Bridges are shown in red.}
}
  	\label{fig:bridges1}
	\end{subfigure}
	\hfil
	\begin{subfigure}[t]{\columnwidth}
    		\centering
      		\includegraphics[height=5.5cm]{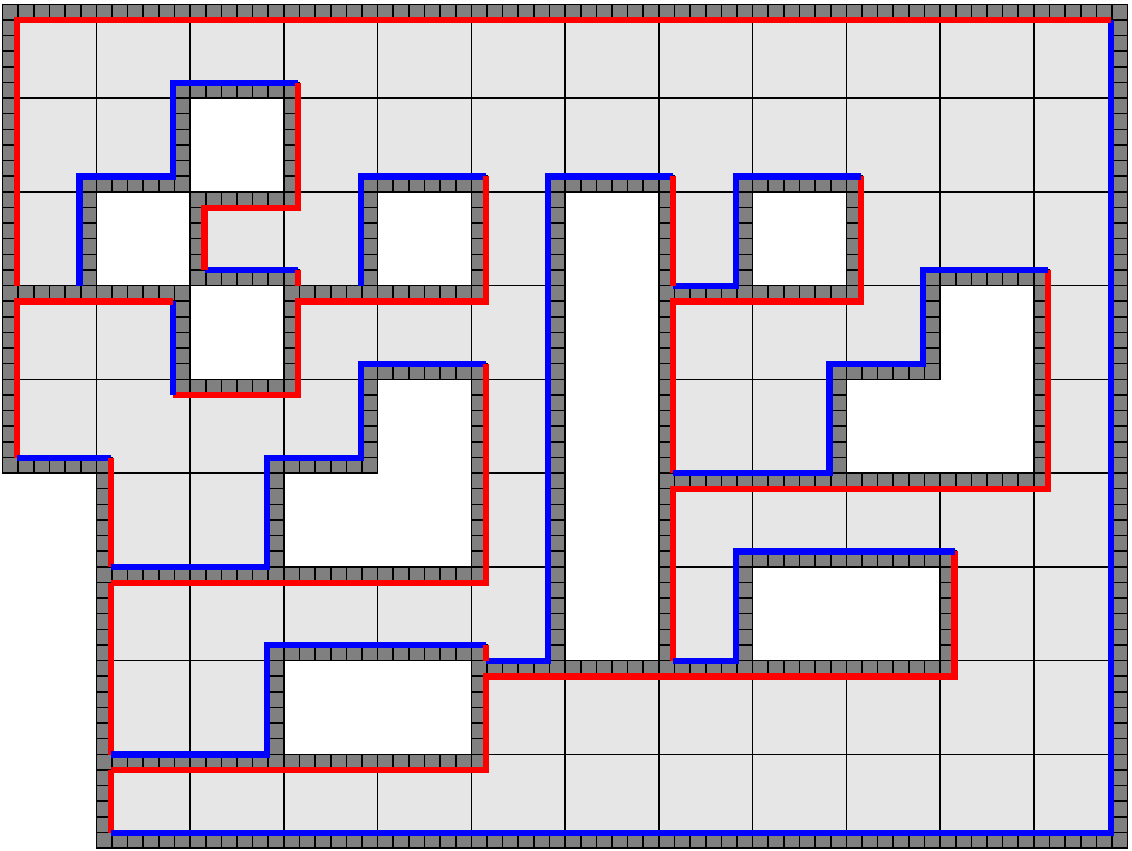} 
    		\caption{\new{The final pieces $S_1$ (dark grey) and $S_2$ (light grey), whith blue and red glue along the common boundary.}
}
  	\label{fig:bridges2}
	\end{subfigure}
    \caption{\new{Partitioning $P^6$ into $S_1$ and $S_2$.}
}
  \label{fig:bridges}
\end{figure}

\begin{proof}
From a high-level point of view, the approach constructs two supertiles
$S_1$ and $S_2$ separately and finally glues them together, see Figure~\ref{fig:bridges} for the overall construction. 
The first supertile
$S_1$ consists of the boundaries of all holes, the boundary of the whole
polyomino, and connections between these boundaries. 
The second supertile $S_2$ is composed of the rest of the polyomino. The
scale factor of 6 guarantees that \new{after removing the boundary pieces connected into $S_1$,}
$S_2$ is hole-free \new{and has a thickness of 4}, which 
allows employing the approach of Theorem~\ref{th:nohole1} \newest{in the version of Corollary~\ref{cor:thick}.}

	{\noinbf{\new{Partition} of $P^6$ into $S_1$ and $S_2$}:}
\new{$S_1$ consists of all boundary pixels of $P^6$, connected by additional connections (``bridges''),
such that the remainder $S_2=P^6\setminus S_1$ is connected and with vertical thickness four.
To this end, consider the set of connected components of boundary pixels of $P^6$; one of them
(say, $C_0$) contains the outside boundary, while the inside components (say, $C_1,\ldots,C_k$) surround holes. 
Because of the scaling, every  boundary pixel has $x$-cooordinate 0 or 1 (modulo 6) or $y$-cooordinate 0 or 1 (modulo 6).
For each inside component $C_i$, consider its bottommost $p_i$ of the leftmost pixels. Because of the scaling, 
its $y$-coordinate is 1 modulo 6. 
For each inside component $C_i$, we add pixels to the left of $p_i$ until we 
hit a pixel of another boundary, say, $C_{j_i}$,
thereby building a {\em bridge} $B_i$ from $C_i$ to $C_{j_i}$. Considering the coordinates of boundary pixels modulo 6, we
conclude that no pixel of a bridge $B_i$ can be adjacent to a pixel of a boundary component other than $C_i$ and $C_{j_i}$.
Therefore, the set of bridges induces a directed tree, with $C_0$ as the root \newest{node}. As a consequence, the complement
$S_2=P^6\setminus S_1$ is hole-free. Furthermore, the horizontal pieces of $S_1$ have $y$-coordinates
0 and 1 modulo 6, so $S_2$ has vertical thickness at least 4: Any pixel at the lower end of a boundary has
$y$ coordinate 1 modulo 6, while a pixel at the upper end of a boundary has $y$ coordinate 0 modulo 6,
so any vertical cut through $S_2$ must contain at least four pixels with $y$-coordinates
2, 3, 4, 5 modulo 6.}


	\newnewtext{\noinbf{Assembling $S_2$}.} \new{Because $S_2$ is hole-free and has vertical thickness 4,}
we can apply the approach of
Theorem~\ref{th:nohole1}. During the construction of $S_2$,
we guarantee that each northern and western face of the boundary of $S_2$ is
marked by the $\textit{red}$ glue and each face from the eastern or southern boundary of
$S_2$ by the $\textit{blue}$ glue, see Figure~\ref{fig:bridges2}. To guarantee that
these two additional glues do not interfere with the self-assembly of $S_2$,
$\textit{red}$ and $\textit{blue}$ are not used in the approach
from Theorem~\ref{th:nohole1}. Overall, we
need $\mathcal{O}(\log^2 (n))$ stages, $20$ glues and $\mathcal{O}(k)$ bins 
for the construction of $S_2$.
	
	\newnewtext{\noinbf{Assembling $S_1$}.} \new{By construction of $S_1$, the bridges induce
a tree between the boundary components. Hence, we can again use a median-based decomposition, i.e., 
recursively choose a boundary component that splits
the tree into components that have at most half the numbers of components.
Each boundary $C_i$, in turn, is again split into two chains that both contain at
most half the number of pixels of $C_i$, see Figure~\ref{fig:constrS2}.}
At the cutting pixels, we use two different glues for a unique
attachment. For the staged self-assembly of the remaining chains, we apply the
approach that is used to self-assemble strips in logarithmically many
stages~\cite{DDF08}. In particular, we split each chain recursively in the
middle and mark the cutting pixels by a third glue type that is not used by
both end pixels of the chain, see Figure~\ref{fig:constrS2}. Finally, a fourth
glue is needed for the attachment of a bridge $B_i$ to its parent component $C_i$ in the tree decomposition. 
This bridge, including its exit pixel at the other component $C_{j_i}$, can be constructed by
applying the same approach that is used to construct the boundary components. Analogously
to the construction of $S_2$, we mark the boundary \new{of $S_1$} by $\textit{red}$ and
$\textit{blue}$, but now in the opposite direction: Each northern and western
face of the boundary of $S_1$ is marked by the blue glue; each face from the
eastern or southern boundary of $S_1$ is marked by the red glue, see
Figure~\ref{fig:bridges2} and Figure~\ref{fig:constrS2}. Again, $\textit{red}$ and
$\textit{blue}$ are not allowed to be used for the construction of the
supertile \new{for} $S_1$. Overall, we need $\mathcal{O}(\log^2 (n))$ stages, $6$
glues and $\mathcal{O}(k)$ bins for the construction of
$S_2$.
	
\begin{figure}[ht]
  \begin{center}

       \includegraphics[height=9cm]{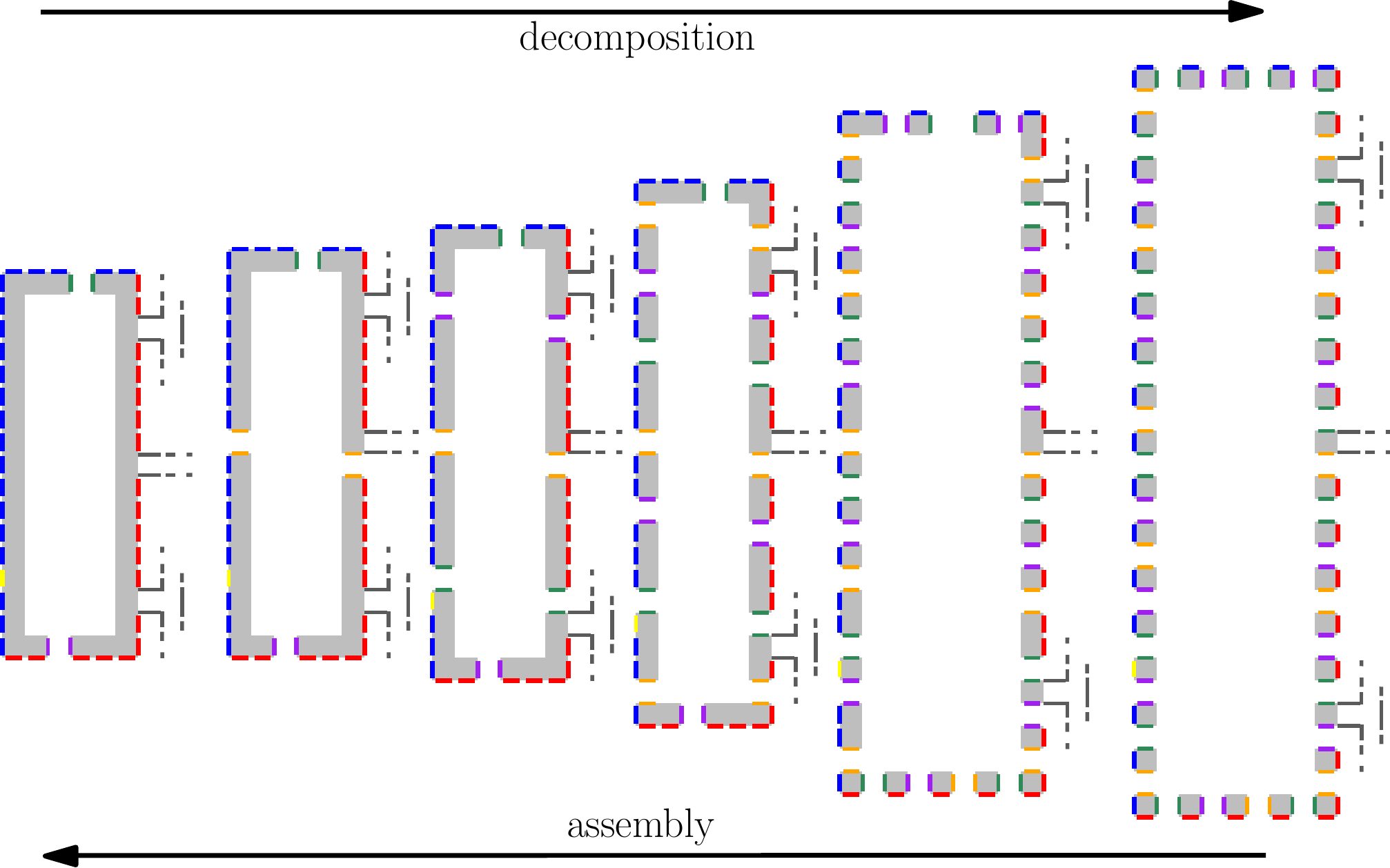}
   
  \end{center}
  \vspace*{-12pt}
  \caption{Recursive separation of a circle from $S_2$.}
  \label{fig:constrS2}
\end{figure}

\new{\noinbf{Putting together $S_1$ and $S_2$}.}
\new{Finally, it is straightforward to see that the geometry of $S_1$ and $S_2$ implies that
they can only attach to each other in the canonical manner, making use of the red and blue glue, as shown in Figure~\ref{fig:bridges2}.}

\new{{\noinbf{Overall glue and stage complexity.}}}
\new{As $S_1$ and $S_2$ are assembled in different bins and
because all these glues bond after the construction of $S_1$ and $S_2$ we are
allowed to make use of 6 glues from the construction of $S_2$ for the construction of
$S_1$. Hence, overall we need $2 + \max \{ 18,6 \} = 20$ glues for the
construction of $P$. Thus, we obtain a total glue complexity of $20$, stage complexity
$\mathcal{O}(\log^2 (n))$, bin complexity $\mathcal{O}(k)$ with a scale factor
of $6$. All \new{bonds} use temperature $\tau=1$.}

	
\end{proof}

\section{Future Work}
Our new methods have the same stage and bin complexity {as previous work \new{on} stage assemblies~\cite{DDF08}} and use just a
small number of glues. Because the bin complexity is in $\mathcal{O}(k)$ for a 
polyomino with $k$ vertices, we may need many bins if the
polyomino has many vertices. Hence, all our methods are excellent for shapes
with a compact geometric description. This still leaves the interesting challenge
of designing a staged assembly system with similar stage, glue and tile complexity,
but a better bin complexity for polyominoes with many vertices, e.g, for $k \in \Omega(n^2)$.

Another interesting challenge is to develop a more efficient 
system for an arbitrary polyomino. Is there a staged assembly
system of stage complexity $o(\log^2 n)$ without increasing the
other complexities? 

\section*{Acknowledgments}
We thank the anonymous reviewers for their patient and constructive approach that greatly helped to improve
the presentation of many aspects of this paper.

\small
\bibliographystyle{abbrv}
\bibliography{literatur}

\end{document}